\documentclass[lettersize,journal]{IEEEtran}
\IEEEoverridecommandlockouts
\usepackage[mathscr]{eucal}
\usepackage[cmex10]{amsmath}
\usepackage{epsfig,epsf,psfrag}
\usepackage{amssymb,amsmath,amsthm,amsfonts,latexsym}
\usepackage{amsmath,graphicx,bm,xcolor,url,overpic}
\usepackage{subfigure}
\usepackage{fixltx2e}
\usepackage{array}
\usepackage{verbatim}
\usepackage{bm}
\usepackage{bbm}
\usepackage{algorithmic}
\usepackage{algorithm}
\usepackage{verbatim}
\usepackage{textcomp}
\usepackage{mathrsfs}
\usepackage{epstopdf}

\newcommand{\openone}{\leavevmode\hbox{\small1\normalsize\kern-.33em1}}

\catcode`~=11 \def\UrlSpecials{\do\~{\kern -.15em\lower .7ex\hbox{~}\kern .04em}} \catcode`~=13 

\allowdisplaybreaks[1]

\newcommand{\calA}{\mathcal{A}}
\newcommand{\calB}{\mathcal{B}}

\newcommand{\calF}{\mathcal{F}}

\newcommand{\calP}{\mathcal{P}}

\newcommand{\calS}{\mathcal{S}}

\newcommand{\calX}{\mathcal{X}}

\newcommand{\bA}{\mathbf{A}}

\newcommand{\bB}{\mathbf{B}}

\newcommand{\rma}{\mathrm{a}}
\newcommand{\rmA}{\mathrm{A}}
\newcommand{\rmb}{\mathrm{b}}

\newcommand{\rmd}{\mathrm{d}}
\newcommand{\rmD}{\mathrm{D}}

\newcommand{\bbE}{\mathbb{E}}

\newcommand{\bbN}{\mathbb{N}}

\DeclareMathAlphabet{\mathbsf}{OT1}{cmss}{bx}{n}
\DeclareMathAlphabet{\mathssf}{OT1}{cmss}{m}{sl}

\DeclareSymbolFont{bsfletters}{OT1}{cmss}{bx}{n}  
\DeclareSymbolFont{ssfletters}{OT1}{cmss}{m}{n}
\DeclareMathSymbol{\bsfGamma}{0}{bsfletters}{'000}
\DeclareMathSymbol{\ssfGamma}{0}{ssfletters}{'000}
\DeclareMathSymbol{\bsfDelta}{0}{bsfletters}{'001}
\DeclareMathSymbol{\ssfDelta}{0}{ssfletters}{'001}
\DeclareMathSymbol{\bsfTheta}{0}{bsfletters}{'002}
\DeclareMathSymbol{\ssfTheta}{0}{ssfletters}{'002}
\DeclareMathSymbol{\bsfLambda}{0}{bsfletters}{'003}
\DeclareMathSymbol{\ssfLambda}{0}{ssfletters}{'003}
\DeclareMathSymbol{\bsfXi}{0}{bsfletters}{'004}
\DeclareMathSymbol{\ssfXi}{0}{ssfletters}{'004}
\DeclareMathSymbol{\bsfPi}{0}{bsfletters}{'005}
\DeclareMathSymbol{\ssfPi}{0}{ssfletters}{'005}
\DeclareMathSymbol{\bsfSigma}{0}{bsfletters}{'006}
\DeclareMathSymbol{\ssfSigma}{0}{ssfletters}{'006}
\DeclareMathSymbol{\bsfUpsilon}{0}{bsfletters}{'007}
\DeclareMathSymbol{\ssfUpsilon}{0}{ssfletters}{'007}
\DeclareMathSymbol{\bsfPhi}{0}{bsfletters}{'010}
\DeclareMathSymbol{\ssfPhi}{0}{ssfletters}{'010}
\DeclareMathSymbol{\bsfPsi}{0}{bsfletters}{'011}
\DeclareMathSymbol{\ssfPsi}{0}{ssfletters}{'011}
\DeclareMathSymbol{\bsfOmega}{0}{bsfletters}{'012}
\DeclareMathSymbol{\ssfOmega}{0}{ssfletters}{'012}

\newcommand{\balpha}{\bm{\alpha}}

\newcommand{\veps}{\varepsilon}

\newcommand{\blambda}{\bm{\lambda}}

\DeclareMathOperator*{\argmin}{arg\,min}

\usepackage{thmtools, thm-restate}
\newtheorem{theorem}{Theorem} 
\newtheorem{lemma}[theorem]{Lemma}

\newtheorem{proposition}[theorem]{Proposition}

\newtheorem{definition}{Definition}

\newtheorem{remark}{Remark}

\usepackage{tikz}
\graphicspath{{./figures/}}
\usepackage{cite}
\usepackage{amsmath,amssymb,amsfonts}
\usepackage{algorithmic, epstopdf}
\usepackage{graphicx}
\usepackage{textcomp}
\usepackage{xcolor}
\usepackage{hyperref}
\usepackage{multirow}
\usepackage{threeparttable}

\usepackage{mathtools}
\mathtoolsset{showonlyrefs}

\hypersetup{
    colorlinks,
    linkcolor = {red!60!black},
    citecolor = {blue!60!black},
    urlcolor = {blue!50!black}
}
\def\BibTeX{{\rm B\kern-.05em{\sc i\kern-.025em b}\kern-.08em
    T\kern-.1667em\lower.7ex\hbox{E}\kern-.125emX}}
\begin{document}

\title{Asymptotic Nash Equilibrium for the $M$-ary Sequential Adversarial Hypothesis Testing Game
}

\author{\IEEEauthorblockN{Jiachun Pan$^*$, Yonglong Li$^*$, Vincent Y.~F.~Tan$^{*,\dagger}$}   
\thanks{This work was presented in part at the 2022 International Symposium on Information Theory (ISIT) in Espoo, Finland.}
	
\vspace{0.5em}

\IEEEauthorblockA{$^*$Department of Electrical and Computer Engineering, National University of Singapore \\
$^\dagger$Department of Mathematics, National University of Singapore
}

\IEEEauthorblockA{Emails: \url{pan.jiachun@u.nus.edu}, \url{elelong@nus.edu.sg}, \url{vtan@nus.edu.sg}}
}

\maketitle

\begin{abstract}
In this paper, we consider a novel $M$-ary sequential hypothesis testing problem in which an adversary is present and perturbs the distributions of the  samples before the decision maker observes them. This problem is formulated as  a sequential adversarial hypothesis testing game played between the decision maker and the adversary. This game  is a zero-sum and strategic one. We assume the adversary is active under \emph{all} hypotheses and knows the underlying distribution of observed samples. We adopt this framework as it is the worst-case scenario from the perspective of the decision maker. The goal of the decision maker is to minimize the expectation of the stopping time to ensure that  the test is as  efficient as possible; the adversary's goal is, instead, to maximize the stopping time.  We derive a pair of strategies under which the asymptotic Nash equilibrium of the game is attained. We also consider the case in which the adversary is not aware of the underlying hypothesis and hence is constrained to apply the same strategy regardless of which hypothesis is in effect.  Numerical results corroborate   our theoretical findings.
\end{abstract}


\section{Introduction}
Hypothesis testing is a fundamental problem in statistics and information theory. There are many works that have laid  firm theoretical foundations for the fundamental limits of hypothesis testing. In this paper, we consider a new setting in which there exists an adversary that  deliberately  acts in a malicious way to cause  a sequential test implemented by a decision maker to fail~\cite{barni13}. We term this new setting as the \emph{adversarial sequential hypothesis testing game}.

We are motivated by  security and trustworthy issues of modern machine learning algorithms. Such issues  have been studied extensively in the past decade. Machine learning algorithms can be shown to be highly vulnerable to adversarial perturbations~\cite{adversary}. For example, in image classification problem, there may be adversarial samples that adversely affect  the performance of  classification tasks. The adversary may  adopt different attack strategies for images in different classes. In this case, it is important to identify the true class of images even under the perturbation of the adversary. In all these examples, when the distributions of observed samples are known, we can consider this problem under a game-theoretic framework and formulate the problem  as a hypothesis testing game played by the decision maker and the adversary. 

\subsection{Related Works}

The works that are closely related to the present paper are those by Barni and Tondi~\cite{barni13,barni14,barni18,Barniwhole}. 
In these works, the authors considered a general framework to analyze   binary hypothesis testing   by taking into account the presence of an adversary who aims to impede the making of a correct decision. They introduced and analyzed an adversarial version of the Neyman--Pearson setup in which a  defender and an adversary face off against each other.  Given a null hypothesis $H_0$ and a test sequence $Z^n$, the defender must decide whether to accept hypothesis $H_0$ characterized by a distribution $P_0$. As in the classical Neyman--Pearson scenario, the defender must ensure that the type-I error probability (i.e., the probability of rejecting $H_0$ when $H_0$ is true) is no larger than a  prescribed constant $\alpha\in (0,1)$. In turn, the adversary observes a sequence $Y^n$ generated under an alternative hypothesis $H_1$, characterized by a different distribution $P_1$, and transforms it into a modified sequence so that when presented with the modified sequence, the defender still accepts $H_0$. In other words, the adversary aims at maximizing the type-II error probability (i.e., the probability that the defender accepts $H_0$ when $H_1$ holds), while the defender's goal is to minimize it by taking into account the presence of the adversary. In the  setting of~\cite{barni13,barni14,barni18,Barniwhole}, adversarial hypothesis testing is modeled as a zero-sum game. In~\cite{barni13}, the authors consider the case in which $P_0$ and $P_1$ are both known to the defender and the adversary. They showed that under certain assumptions, the game admits asymptotic Nash equilibrium and obtained the optimum strategies for the decision maker and the adversary at the equilibrium. In~\cite{barni14}, the authors extended their previous works by considering a scenario in which $P_0$ is known only through one or more training sequences. They also derive the asymptotic Nash equilibrium of this setting. In~\cite{barni18}, the authors also assume $P_0$ is known through training sequences but in our paper, the training data is   corrupted by an adversary. 

While~\cite{barni13,barni14,barni18} characterize the adversarial hypothesis testing problem when the adversary is only active in {\em one of the two} hypotheses, it is also reasonable to consider the case when the adversary is active under \emph{all} hypotheses. Tondi, Barni, and Merhav~\cite{tondi2015detection} extended the game-theoretic formulation of the defender-adversary interaction to the case where the attacker acts under both hypotheses. Under this setting, a dominant (i.e., optimal regardless of what the defence strategy is) and universal (i.e., not dependent on  the underlying sources) adversary strategy can be obtained. Furthermore,  Jin and Lai~\cite{yulu} also focus on this setting but they formulated it as a minimax problem. They obtain a nonasymptotic saddlepoint solution which reveals the optimal attack and defense strategies.  

Instead of directly perturbing the observed   sequence of samples, there are also works that permit the  adversary to perturb the underlying {\em distributions}. In Yasodharan and Loiseau~\cite{Yasodharan}, the adversary chooses any distribution from a set of distributions and assigns each choice of distribution a cost function. Then they considered non-zero-sum hypothesis testing games in both the Bayesian and the Neyman--Pearson frameworks. The authors showed that these games admit mixed strategy Nash equilibra. Zhang and Zou~\cite{ruizhi}  extended the non-zero-sum hypothesis testing games in~\cite{Yasodharan} to the sequential case and obtain the asymptotic Nash equilibrium. They  first guessed the strategy $s_{\rma}$ that the adversary adopts and then designed a strategy $s_{\rmd}(s_{\mathrm{a}})$ of the decision maker based on adversary's strategy. However, their methods cannot be extended to the case when the adversary is active in all hypotheses.

Another line of work that is similar to our setting is the \emph{robust hypothesis testing problem}. A robust binary hypothesis test is a  \emph{minimax} test for two hypotheses where the actual probability distributions of the observations are located in neighborhoods of a nominal density. The actual and nominal distributions are constrained in terms of a certain distance measure such as the relative entropy~\cite{RHT}, the $\alpha$-divergence~\cite{gul2016robust}, and the Wasserstein distance~\cite{gao2018robust}. The above-mentioned works show that the minimax solution is an optimal test based on  the \emph{least favorable distributions} (LFDs), i.e., a test that optimally separates the closest feasible distributions.
For the $M$-ary case,  Fau{\ss}, Zoubir, and Poor~\cite{Fau}  considered a sequential $M$-ary robust hypothesis testing problem. They showed that the minimax solution is also an optimal test for the LFDs, but now the LFDs depend on the previous observations. This results in the  sequence of samples being no longer i.i.d., but rather being a Markov process. In a follow-up work~\cite{fauss2020minimax}, the same authors obtain sufficient conditions for strict minimax optimality of sequential tests for multiple hypotheses under mild Markov assumptions.  The differences between robust hypothesis testing and our problem are discussed in more detail in Remark~\ref{rmk:robust_HT} in Section~\ref{sec:formulation}.

\subsection{Main Contributions}

In this paper, we focus on the $M$-ary sequential adversarial hypothesis testing game. There are $M$ hypotheses $H_i,i\in[M]$ and they are characterized by $M$ different distributions $P_i,i\in[M]$ respectively.  Samples are collected \emph{sequentially} and they are perturbed by the adversary. For the most part of the paper, we assume the adversary knows the underlying distribution of the observed samples, and has the ability to perturb the samples {\em based on which hypothesis is in effect}. 
There are four distinct contributions in our paper. 
\begin{itemize}
    \item We formulate a sequential adversarial $M$-ary hypothesis testing game and state our objective in terms of finding an asymptotic Nash equilibrium between the player and the adversary. The player's objective is a linear combination of error exponents; this is in contrast to other works in robust hypothesis testing (see Remark~\ref{rmk:robust_HT} for details). The adversary is assumed to be powerful; it knows the true distributions, which hypothesis is in effect, and can perturb the player's observations under {\em both} hypotheses.
    \item We derive  optimal strategies  for the player and the adversary that yield an asymptotic Nash equilibrium for this two-player sequential game using information-theoretic tools. Different from~\cite{ruizhi} in which the decision maker first estimates the strategy $s^*_{\rma}$ that the adversary will adopt, and {\em then} designs its strategy $s_{\rmd}(s^*_{\mathrm{a}})$ based on the estimates, in our work, both strategies are executed {\em simultaneously}.
    \item We discuss the case when the adversary is incognizant of the underlying distributions of the observations. This is a weaker form of the adversary. Even though we are unable to obtain the pair of strategies that achieves the asymptotic Nash equilibrium, we show that the decision maker can achieve larger  error exponents compared to the adversary-aware setting. 
    \item Numerical results corroborate our theoretical findings. Specifically, we show on synthetic and real datasets that the empirical performance of the proposed strategies converge to their promised fundamental limits. 
\end{itemize}


\subsection{Paper Outline}

The rest of the paper is structured as follows.  In Section~\ref{sec:pre}, we introduce some preliminary knowledge on the $M$-ary sequential hypothesis testing and two player games. In Section~\ref{sec:formulation}, we formulate the $M$-ary sequential adversarial hypothesis testing problem formally and introduce the definition of asymptotic Nash equilibrium.  In Section~\ref{sec:main}, we present our main theorem (Theorem~\ref{thm:advboth}) about the set of strategies at which the asymptotic Nash Equilibrium can be obtained and the proof of our main theorem. In Section~\ref{sec:nonawareness}, we consider a weaker form of the adversary who does not know the underlying distributions and derive bounds on the performance of the decision maker. In Section~\ref{sec:example}, we provide some numerical simulations. We conclude the paper in Section~\ref{sec:con} and propose some directions for future researches.

\section{Preliminaries}
\label{sec:pre}

\subsection{$M$-ary Sequential Hypothesis Testing}
In this section, we discuss the $M$-ary sequential hypothesis testing setup~\cite{MSHT}. Let $\{ X_i \}_{i=1}^\infty$ be a  sequence of independent and identically distributed (i.i.d.)  random variables with distribution $P$, and let $H_i$ be the hypothesis that $P=P_i$ for $i=1,2,\ldots,M$. We assume that $P_i\neq P_j$ for all $i \ne j$. The objective of this problem is to uncover the true hypothesis with a desired accuracy as quickly as possible (i.e., using the fewest number of samples). In this problem, there is a fundamental  tradeoff between the number of samples and the error probabilities. 

We will use {\em sequential} tests to learn the underlying hypothesis. Such tests consist  of  stopping rules and  final decision rules. The stopping rule determines the number of samples that are collected until a decision is made and the final decision rule decides which of the $M$ hypotheses is the true one. 

For $n\geq 1$, we define the {\em  log-likelihood ratio} between distributions $P_i$ and $P_j$ as
\begin{align*}
    S_{ij}(n)=\sum_{k=1}^n\log\frac{P_i(X_k)}{P_j(X_k )}.
\end{align*}
For a {\em threshold} or {\em boundary matrix} $\bB=[B_{ij}]$, with $B_{ij}>0$ and the $B_{ii}=0$, define the \emph{matrix sequential probability ratio test} (MSPRT)~\cite{tartakovsky2014sequential} $\delta^*_M=(T_M^*,d_M^*)$ that is constructed based on   $(M+1)M/2$ one-sided SPRTs between hypotheses $H_i$ and $H_j$  as follows: The stopping rule is 
\begin{align*}
    \mbox{Stop at the first } n\!\ge\! 1~\mbox{s.t. } \exists \, i\in [M]~\mbox{s.t. }  S_{ij}(n)\geq B_{ij}~\forall\, j\!\ne\! i.
\end{align*}
Accept the unique $i \in [M]$ that satisfies these inequalities. Note that for $M=2$ this test coincides with Wald's   sequential probability ratio test (SPRT)~\cite{wald1948}. It can be shown that the MSPRT with proper thresholds is first-order asymptotically optimal in the sense of minimizing the expected sample sizes for all hypotheses~\cite[Chapter~4.3]{tartakovsky2014sequential}, i.e. for all tests $\delta_M=(T_M,d_M)$ with all error probabilities upper bounded by $\alpha_{\max} \in (0,1)$, we have
\begin{align*}
    \lim_{\alpha_{\max}\to 0}\inf_{\delta_M} \bbE_i[T_M]= \bbE_i[T^*_M],~\mbox{for all } i=1,2,\ldots,M,
\end{align*}
where $T_M^*$ is the optimal stopping time. 

\subsection{Two Player Games}

We now provide a brief introduction to two-player games. For a more detailed exposition, the reader is referred to~\cite{osborne2004introduction}. A {\em two-player game} is defined as a quadruple $(\calS_1,\calS_2,u_1,u_2)$, where $\calS_1$, $\calS_2$ are the sets of strategies (actions) the first and the second player can choose from, and $u_i (s_1,s_2)$ (where $s_1\in\calS_1$ and $s_2\in\calS_2$) is the \emph{payoff} (i.e., the gain) for player $i\in\{1,2\}$, when the first player chooses the strategy $s_1\in\calS_1$ and the second choose $s_2\in\calS_2$. A pair of strategies $(s_1,s_2)$ is called a \emph{profile}. In a zero-sum competitive game, the sum of the two payoffs is equal to $ 0$, i.e., $u_1(s_1, s_2)+u_2(s_1,s_2)=0$ for all $(s_1,s_2)\in\calS_1\times\calS_2$. In other words, the gain of a player is equal to the loss of the other. We define the payoff function for a zero-sum game as $u=u_1=-u_2$. A \emph{strategic} game is a model of interaction in which each player chooses an action not having been informed of the other player's action. We can think of the players' action as being taken ``simultaneously''. One common goal is to obtain a {\em Nash equilibrium}~\cite{Nash48} of a zero-sum, strategic game, which is defined as follows. A profile $(s_1^*, s_2^*)$ is a {\em Nash equilibrium} if:
\begin{alignat*}{2}
 u(s_1^*, s_2^*)&\geq u(s_1, s_2^*)\quad &&\forall s_1\in\calS_1,\quad\mbox{and}\\
u(s_1^*, s_2^*)&\leq u(s_1^*, s_2)\quad && \forall s_2\in\calS_2.
\end{alignat*}
In other words, a profile is a Nash equilibrium if no player can increase his/her payoff by changing his/her strategy unilaterally.

\section{Problem Formulation}
\label{sec:formulation}
In this section, we first formulate the sequential $M$-ary adversarial hypothesis testing game.  Let ${\cal X}=\{a_1, a_2, \ldots, a_K\}$ be the finite alphabet of the source and $\calP (\calX)$ be the set of probability mass functions (also called distributions) supported on $\calX$.  There are $M$ hypotheses. We use $[M]$ to denote the   finite set $\{1, \ldots, M\}$. Under hypothesis $H_i$, the underlying distribution is $P_i$ for $i \in [M]$. We also assume that the distributions $P_i, i\in[M]$ are  known to both the decision maker and the adversary. Here the adversary perturbs the distribution; this has the effect of passing the original samples $\{X_i\}_{i=1}^\infty$ through a discrete memoryless channel, which we denote by $\bA$ where the entries $[\bA]_{lj}=\Pr(Y=a_j|X=a_l)$ for $l,j\in [K]$. We assume that the channels $[\bA]$ are chosen such that $\sum_{l=1}^K P_i(X=a_l)[\bA]_{lj}>0$ for $j\in [K]$ and $i\in[M]$, which means that the distribution of $Y$ (i.e., $Y\sim P_i\bA, i\in[M]$) has full support.  Besides, motivated by the fact that the adversary's power is bounded, we  impose a distance constraint between the input distribution and output distribution of the adversary. This is  characterized by distance measure/metric $d$. Here we do not specify the choice of measure $d$ for now and we aim to  obtain results  under some specific conditions on $d$. Then the adversary's constraint is 
\begin{align*}
d(P_i,P_i\bA)\leq \Delta, \qquad \forall\, i\in[M],
\end{align*}
where $\Delta>0$ is a prescribed maximum distance between $P_i$ and $P_i\bA$ and $\Delta$ should be small to ensure that $\min_{i,j\in[M],i\neq j}D(P_i\bA_i\|P_j\bA_j)\ge\epsilon>0$. This means the KL divergences between perturbed distributions is positive. This constraint is important as it ensures that the true hypothesis can be learned uniquely. Besides, we assume $\Delta$ is known to the decision maker. 

When the adversary is active under all hypotheses, there are two different scenarios we can consider. Firstly, the adversary knows underlying hypothesis $H_i, i\in[M]$ and secondly, the adversary does not. For the majority of the paper, we  consider the {\em awareness} case as it is the worst-case scenario from the perspective of the decision maker. Later in Section~\ref{sec:nonawareness}, we discuss the {\em non-awareness} case, i.e., the adversary is not aware the underlying distribution.

Fig.~\ref{fig:actboth} shows the $M$-ary sequential adversarial hypothesis testing game when the adversary knows the underlying distribution of $X$. At each time $n\in\bbN$, a sample $X_k$ is generated from $P_i$ and given to the adversary. The adversary modifies $X_k$ to $Y_k$ using the attack strategy. Here we note that the adversaries are different for $H_i,i\in[M]$. We denote the adversary's strategy/channel under $H_i$ as $\bA_i$ for $i\in[M]$. Based on the adversary strategy, the distribution of $Y$ is $P_i\bA_i,i\in[M]$.   Then the objective of decision maker is to decide which hypothesis is true based on the sequence up to the current time $\{Y_k\}_{k=1}^n$.

\begin{figure}[t]
	\centering
	\includegraphics[width=\linewidth]{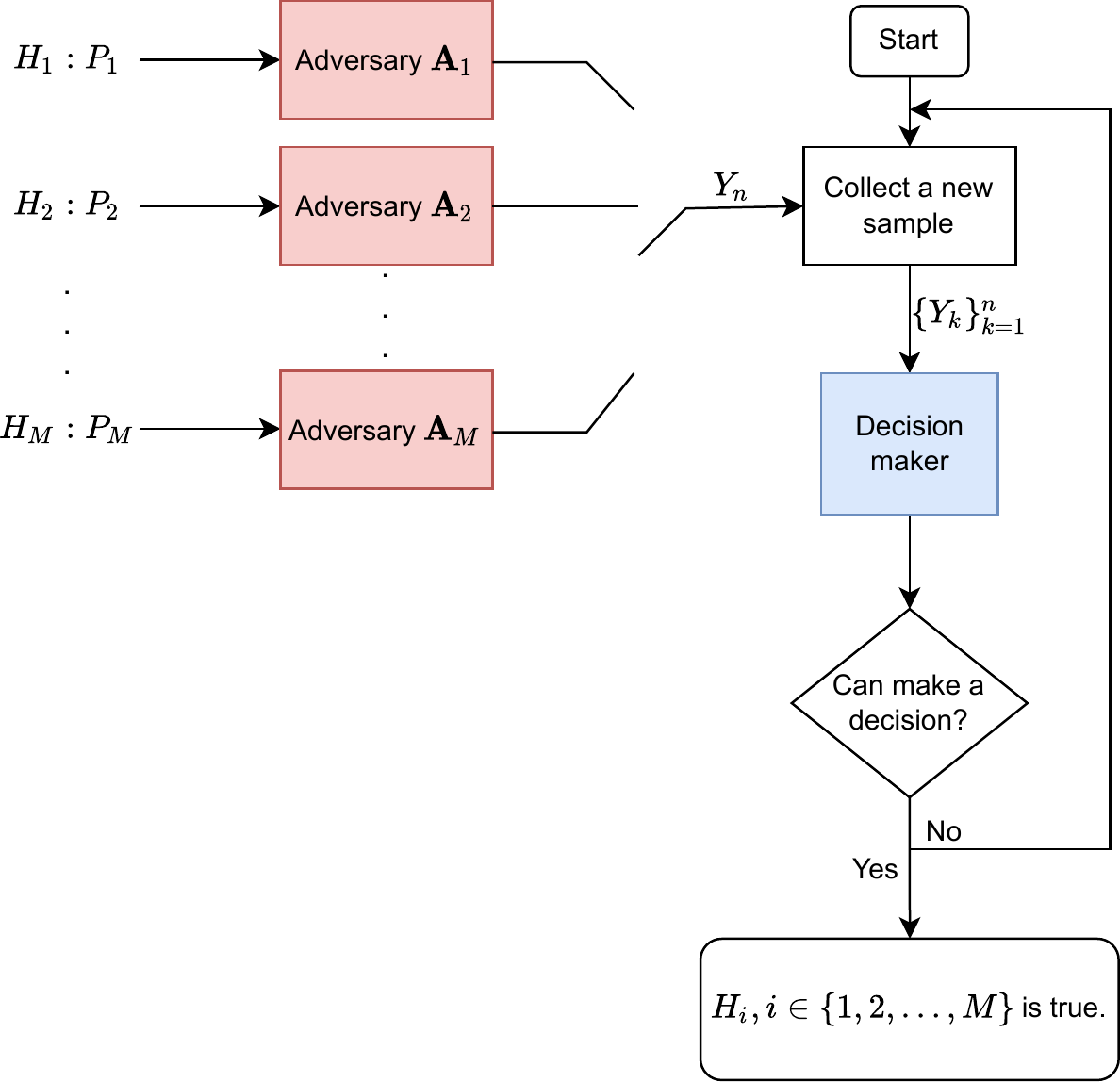}
	\caption{Illustration of sequential adversarial hypothesis testing when the adversary is active in both hypotheses}
	\label{fig:actboth}
\end{figure}

We define the integer-valued random variable $T\in\bbN$ as the stopping time with respect to the filtration $\{\calF_n=\sigma(Y_1, Y_2,\ldots,Y_n)\}_{ n\in\bbN}$ generated by the  samples up to time $n$. To achieve the goal, the decision maker  at each time $n$ can take one of two actions:
\begin{itemize}
\item Stop drawing a new sample and declare that one of $H_i,i\in[M]$ is true.
\item Continue to draw a new sample.
\end{itemize}
We denote the expectation of the stopping time under $H_i$ as $\bbE_i[T]$ for $i\in[M]$. The decision rule $\delta$ is a $M$-valued $\calF_{T}$-measurable function.  A test is a pair $\Phi=(T,\delta)$. To avoid trivialities, we only consider the tests with finite expected  stopping time (i.e., $\bbE_i[T]<\infty$ for all $i\in[M]$) in the sequel.

Let $\alpha_{ij}(\Phi,(\bA_1,\ldots,\bA_M))=P_i(\delta=j)$ for $i\neq j$ be the error probabilities of the test $\Phi$, i.e., the probabilities of accepting a specific $H_j$ when $H_i$ is true. Then $\alpha_i(\Phi,(\bA_1,\ldots,\bA_M))=P_i(\delta\neq i),i\in[M]$ is the probability of rejecting   hypothesis $H_i$ and $\alpha_i=\sum_{j\neq i}\alpha_{ij}$. We denote $\balpha=(\alpha_1,\alpha_2,\ldots,\alpha_M)$ as the error probability vector.
We now provide a formal definition of the \underline{$M$}-ary \underline{SEQ}uential \underline{A}dversarial \underline{H}ypothesis \underline{T}esting problem in the adversary awareness case. We denote this problem as MSEQ-AHT($\calS_{\rmD}(\alpha),\calS_{\rmA}(\Delta),u_{\blambda}^{(\alpha)})$.
\begin{definition}
The MSEQ-AHT($\calS_{\rmD}(\alpha),\calS_{\rmA}(\Delta),u_{\blambda}^{(\alpha)})$ is a zero-sum, strategic, game played by the decision maker and the adversary, defined by the following strategies and payoff.
\begin{itemize}
\item The admissible set of strategies that the adversary can choose from is  
\begin{align}
\label{eqn:stratA}
\!\!\calS_{\rmA}(\Delta)\!:=\!\bigg\{\!(\bA_1,\ldots,\bA_M):\! \max_{i\in[M]}\! d(P_i,P_i\bA_i)\!\leq\! \Delta \bigg\}.
\end{align} 

\item The set of  strategies that the decision maker can choose from is
\begin{align}
\calS_{\rmD}(\alpha):=\bigg\{\Phi:\max_{i\in[M]}\sup_{(\bA_1,\ldots,\bA_M)\in\calS_A(\Delta)}\alpha_i(\Phi)\leq \alpha 
\bigg\}.
\end{align}

\item The payoff when the strategy of decision maker is $\Phi$ and the strategy of the adversary is $(\bA_1,\ldots,\bA_M)$ is
\begin{align}
u_{\blambda}^{(\alpha)}(\Phi,(\bA_1,\ldots,\bA_M))=\sum_{i=1}^M\lambda_i\frac{\log\frac{1}{\alpha}}{\bbE_i[T]},
\end{align}
where $\blambda=(\lambda_1,\ldots,\lambda_M)$ is a vector with positive elements (weights)  that reflects the relative importances of the expected stopping times  $\bbE_i[T]$ for all $i\in[M]$.
\end{itemize}
\end{definition}

In the definition of MSEQ-AHT($\calS_{\rmD}(\alpha),\calS_{\rmA}(\Delta),u_{\blambda}^{(\alpha)})$, the set of strategies for the adversary is comprised of all transition matrices that satisfy the distortion constraints. The set of strategies for the decision maker is comprised of all tests that the test error probabilities are upper bounded by a common $\alpha$. Besides, the payoff is  a linear combination
of the error exponents of the error probabilities $\alpha_i,i\in[M]$, and the decision maker wants to maximize it to make the detection more accurate and efficient, while the adversary wants to minimize it. For MSEQ-AHT($\calS_{\rmD}(\alpha),\calS_{\rmA}(\Delta),u_{\blambda}^{(\alpha)})$, our goal is to obtain a profile $(\Phi^*,(\bA_1^*,\ldots, \bA_M^*))$ that achieves the {\em asymptotic Nash equilibrium} as $\alpha$ tends to zero, which is defined as follows.
\begin{definition}[Asymptotic Nash Equilibrium]
We say that the (family of) profile(s) $(\Phi^*,(\bA_1^*,\ldots, \bA_M^*))$  (indexed by $\alpha>0$) satisfies the asymptotic Nash equilibrium as $\alpha\to 0^+$ if
\begin{align}
&\lim_{\alpha\to0^+} u_{\blambda}^{(\alpha)}(\Phi^*,(\bA_1^*,\ldots, \bA_M^*))\notag\\
&\;\;\geq \lim_{\alpha\to0^+} \sup_{\Phi\in\calS_{\rmD}(\alpha)}u_{\blambda}^{(\alpha)}(\Phi,(\bA_1^*,\ldots, \bA_M^*)),\label{eqn:opt1}
\end{align}
and
\begin{align}
&\lim_{\alpha\to0^+} -u_{\blambda}^{(\alpha)}(\Phi^*,(\bA_1^*,\ldots, \bA_M^*))\notag\\
&\;\;\geq \!\lim_{\alpha\to0^+}\!\sup_{(\bA_1,\ldots,\bA_M)\in\calS_{\rmA}(\Delta)}\! -u_{\blambda}^{(\alpha)}(\Phi^*,(\bA_1,\ldots, \bA_M))\label{eqn:opt2}.
\end{align}
\end{definition}


\begin{remark}
Our problem setting is similar to that of the sequential composite hypothesis testing~\cite{pan2022} framework in which   samples are  generated i.i.d.\ by a distribution from a known set of distributions. However, in our work, the set of distributions is determined by the adversary and there is a payoff function that controls the choice of the strategies the adversary and the decision maker.
\end{remark}

\begin{remark} \label{rmk:robust_HT}
Our problem is, however, different from \emph{robust hypothesis testing}~\cite{RHT}. In robust hypothesis testing problems, the true probability distributions are located in the neighborhoods of a nominal distribution. Instead,  in our setting, we assume that the actual distribution is formed by the adversary's perturbation by transition matrices. Besides, in minimax $M$-ary sequential hypothesis tests, e.g., in \cite{fauss2020minimax}, the authors  typically consider finding a sequential test $\Phi$ that minimizes the maximum of expectation of the stopping times over different distributions with the constraints that the  error probabilities $\alpha_i,i\in[M]$ are upper bounded by fixed constants $\bar{\alpha}_i\in[0,1]$ for all $i\in[M]$, i.e.,
\begin{align*}
    \min_{\Phi}\max_{i\in[M]}\bbE_i[T]\quad 
    \mbox{s.t. }\quad \max_{i\in[M]} \alpha_i\leq \bar{\alpha}_i.
\end{align*}
In our problem setting, we consider a linear combination of the exponents $\frac{\log(1/\alpha)}{\bbE_i[T]},i\in[M]$ as the decision maker's payoff function. 
\end{remark}

\section{Main Results}
\label{sec:main}

To obtain the asymptotic Nash equilibrium of MSEQ-AHT($\calS_{\rmD}(\alpha),\calS_{\rmA}(\Delta),u_{\blambda}^{(\alpha)})$, we first propose   strategies for the decision maker and  adversary. Then we prove that this pair of strategies achieves the asymptotic Nash equilibrium.

Define the \emph{type} or \emph{empirical distribution} of the sequence $x^m \in \calX^m$ as
\begin{align*}
\hat{Q}_{x^m}(a):=\frac{1}{m}\sum_{i=1}^m \mathbbm{1}\{x_i=a\},\quad \forall\, a\,\in\calX.
\end{align*}
Denote $\calA_i(\Delta):=\{\bA_i:d(P_i,P_i\bA_i)\leq\Delta\}$ for $i\in[M]$. Then $\calS_{\rmA}(\Delta)=\calA_1(\Delta)\times\ldots\times\calA_M(\Delta)$. For simplicity,   we abbreviate $\calA_i(\Delta)$ as $\calA_i$ for $i\in[M]$. 
Let $\zeta = 0.85$ from now on.\footnote{The constant $0.85$ for $\zeta$ is somewhat arbitrary; any number in $(0,1)$ works for our analyses. We found that $\zeta=0.85$ works best in our numerical experiments.} Define a threshold  
\begin{align}
\label{eqn:threshold}
\gamma_n:=\frac{\log{\frac{C}{\alpha}}}{n}+ \frac{1}{n^{ \zeta}}+\frac{|\calX|\log(n+1)+\log (M-1)}{n},
\end{align}
where $C:=\sum_{n=1}^{\infty}e^{-n^{1-\zeta}}<\infty$ is a    finite constant. 
Define  $$Z_i^{(n)}:=\min_{ j\in[M],j\neq i}\Big[\min_{\bA_j\in\calA_j}D(\hat{Q}_{Y^n}\|P_j\bA_j)\Big].$$ 
Now we define a stopping time as 
\begin{align}
\label{eqn:stop}
T^{*} =T_{\alpha}^*
:= \inf  \big\{ n\geq 1:  \exists \, i\in[M]~\mbox{ s.t. }
Z_i^{(n)}\geq \gamma_n\big\},
\end{align}
and for $i\in[M]$,
\begin{align*}
T_i:=\inf\big\{n\geq 1~:~Z_i^{(n)}\geq \gamma_n\big\}.
\end{align*}
We also define the decision rule as for $i\in[M]$,
\begin{align}
\label{eqn:decision}
\delta^*:=i\quad \mbox{if}\quad T^*=T_i.
\end{align}
Finally,  define 
\begin{align}
\bA_i^{*}:=\argmin_{\bA_i\in\calA_i} \bigg[\min_{ j\in[M]\setminus\{ i\}}\Big[\min_{\bA_j\in\calA_j}D(P_i\bA_i\|P_j\bA_j)\Big]\bigg].\label{eqn:opta0}
\end{align}
We note that $\bA_i^*$ may not be unique.

Then the test used by the decision maker is $\Phi^{*}=(T^{*},\delta^{*})$.  Now we have the following theorem:
\begin{theorem}
\label{thm:advboth}
If $\calS_{\rmA}(\Delta)$ is a compact set, then for any $\blambda$ in which all elements are positive, $(\Phi^*,(\bA_1^*,\ldots,\bA_M^*))$ defined in~\eqref{eqn:stop}--\eqref{eqn:opta0} is the profile that attains the asymptotic Nash equilibrium as $\alpha\to 0^+$. Besides, the payoff at the asymptotic Nash equilibrium is 
\begin{align}
    &\lim_{\alpha\to0^+} ~u_{\blambda}^{(\alpha)}(\Phi^*,(\bA_1^*,\ldots, \bA_M^*))\nonumber\\
    &=\sum_{i=1}^M \lambda_i \min_{j\in[M]\setminus\{i\}}\left[\min_{\bA_j\in\calA_j} D(P_i\bA_i^*\|P_j\bA_j)\right].\label{eqn:aware_case}
\end{align}
\end{theorem}
Theorem~\ref{thm:advboth} shows that as $\alpha\to 0^+$, the decision maker cannot increase the payoff function (i.e., the linear combination of error exponents) by changing its strategy $\Phi^*$ without the adversary changing its strategy $(\bA_1^*,\ldots,\bA_M^*)$. This is the implication of~\eqref{eqn:opt1}. Similarly, as $\alpha\to 0^+$, the payoff function cannot be increased by the adversary changing its strategy $(\bA_1^*,\ldots,\bA_M^*)$ when the strategy of decision maker is fixed  to be $\Phi^*$. This is the implication of~\eqref{eqn:opt2}. We can also find that the strategies at the asymptotic Nash Equilibrium is independent of the choice of $\lambda_i$ for all $i \in [M]$.

For  the optimal strategy $(\bA_1^*,\ldots,\bA_M^*)$ of the adversary, they can be obtained by solving the optimization problems in~\eqref{eqn:opta0}. As the KL divergence $D(Q_0\|Q_1)$ is convex in $(Q_0,Q_1)$, if we choose the distance measure $d$ that results in $\calS_{\rmA}$ being convex (such as  the KL divergence), although obtaining a closed-form solution is difficult, we can solve the  optimization problem numerically using  off-the-shelf convex optimization software.

Now we prove Theorem~\ref{thm:advboth}. The proof consists of three parts. The ideas in the first two parts are adopted from the proof in~\cite{veeravalli}, but here we need to verify some technical conditions such as the stochastic equicontinuity of a certain family of random variables. The third part is, to the best of our knowledge, original.
\begin{proof}[Proof of Theorem~\ref{thm:advboth}]
The proof of  Theorem~\ref{thm:advboth} proceeds in three distinct parts.

\emph{Part 1: Proof of $\Phi^*\in\calS_{\rmD}(\alpha)$.}  We need to show that for $i\in[M]$, 
\begin{align}
\label{eqn:type1}
    \sup_{(\bA_1,\ldots,\bA_M)\in\calS_{\rmA}}\alpha_i(\Phi^*,(\bA_1,\ldots,\bA_M))\leq \alpha.
\end{align}
 We first recall a lemma from~\cite[Theorem 11.2.1]{cover2006elements}.
 \begin{lemma}
 \label{lem:concen}
 If $Y_1,Y_2,\ldots,Y_n$ are i.i.d.\ generated according to a distribution $Q$, for any  $\epsilon>0$,  we have
 \begin{align}
 P_0\Big(D(\hat{Q}_{Y^n}\|Q)\geq \epsilon\Big)\leq(n+1)^{|\calX|} e^{-n\epsilon}.
 \end{align}
 \end{lemma}

Now we prove~\eqref{eqn:type1}. For any $i\in[M]$ and any set of strategies of the adversary $(\tilde{\bA}_1,\ldots,\tilde{\bA}_M)\in\calS_{\rmA}$, we have
\begin{align*}
    \alpha_i(\Phi^*,&(\tilde{\bA}_1,\ldots,\tilde{\bA}_M))\\
    &=\sum_{j\neq i}\alpha_{ij}(\Phi^*,(\tilde{\bA}_1,\ldots,\tilde{\bA}_M))\\
    &= \sum_{j\neq i}\sum_{n=1}^{\infty}P_i(Z_i^{(n)}\geq \gamma_n)\\
    &\overset{(a)}{\leq} (M-1)\sum_{n=1}^{\infty}P_i\big(D(\hat{Q}_{Y^{n}}\|P_i\tilde{\bA}_i)\geq \gamma_{n}\big)\\
    &\overset{(b)}{\leq} \sum_{n=1}^{\infty} \frac{\alpha}{C} e^{-n^{1-\zeta}}\\
    &\leq \alpha,
\end{align*}
where $(a)$ is based on the definition of $Z_i^{(n)}$, $(b)$ is based on Lemma~\ref{lem:concen} and the fact that $C=\sum_{n=1}^{\infty} \exp(-n^{1-\zeta})$. Thus,~\eqref{eqn:type1} holds. Hence, we have proved that $\Phi^*\in\calS_{\rmD}(\alpha)$.

 \emph{Part 2: Obtain the payoff function of the test $\Phi^{*}$.} In this part, we want to evaluate the error exponents of the test $\Phi^*$ for any adversary strategies $(\bA_1,\ldots,\bA_M)\in\calS_{\rmA}$. We first fix the adversary strategy as $(\tilde{\bA}_1,\ldots,\tilde{\bA}_M)\in\calS_{\rmA}$ (but the decision maker does not known this). Then we derive the payoff when the decision maker's strategy is $\Phi^*$ and the adversary's strategy is $(\tilde{\bA}_1,\ldots,\tilde{\bA}_M)$, i.e.,
\begin{align*}
    u_{\blambda}^{(\alpha)}(\Phi^*,(\tilde{\bA}_1,\ldots,\tilde{\bA}_M))=\sum_{i=1}^M\lambda_i\frac{\log\frac{1}{\alpha}}{\bbE_i[T^*]}.
\end{align*}
To obtain the limiting payoff function as $\alpha\to 0^+$, our overall strategy is to   first obtain the almost sure limit of $\log(1/\alpha)/T_i^*$ for each  $i\in[M]$ and then go from almost sure convergence to convergence in mean. For the first step, we need to derive some properties of the stopping time $T^*_i, i\in[M]$ when $0<\alpha\le 1$ and $\alpha\to 0^+$, respectively. For the second step, we need prove that the family of random variables $\big\{\frac{T^*}{\log(1/\alpha)}\big\}_{0<\alpha\le 1}$ is uniformly integrable. We start with  a basic lemma. 

\begin{lemma}[Li, Nitinawarat, and Veeravalli~\cite{liyun}]
Let $B(Q_0, Q_1)$ be the Bhattacharyya distance between two distributions $Q_0$ and $Q_1$, i.e., 
$$B(Q_0, Q_1) :=-\log\bigg(\sum_{x\in\calX}Q_0(x)^{1/2}Q_1(x)^{1/2} \bigg).$$
If $Q_0$ and $Q_1$ are fully supported on $\calX$, then
\begin{align*}
2B(Q_0,Q_1)=\min_{ P\in\calP(\calX)}  D(P\|Q_0)+D(P\|Q_1) .
\end{align*}
\end{lemma}
\noindent It holds that $$B^*:=\min_{i\neq j}\bigg[\min_{\bA_i\in\calA_i.\bA_j\in\calA_j}B(P_i\bA_i,P_j\bA_j)\bigg].$$ We have $B^*>0$ as $\min_{i\neq j}D(P_i\bA_i\|P_j\bA_j)>0$ for any $\bA_i\in\calA_i,i\in[M]$, which is the condition stated in the choice of $\Delta$ in Section~\ref{sec:formulation}. We can now control the probability that the stopping time exceeds a certain deterministic value $n$.
\begin{lemma}
\label{lem:finite}
For every $n\geq 1$ and $i\in[M]$, we have
\begin{align*}
P_i(T^{*}\geq n)\leq\frac{1}{\alpha} e^{-(n-1)2B^*} (M-1)n^{2|\calX|} e^{(n-1)^{\zeta}}.
\end{align*}
\end{lemma}
\begin{proof}
Without loss of generality, we consider the $i=1$ case. We have
\begin{align*}
&P_1(T^*\geq n)\\
&\leq P_1\Big(\bigcap_{i=1}^M \left\{Z_i^{(n-1)}\leq \gamma_{n-1}\right\}\Big)\\
&\leq P_1\left(\min_{j\neq 1}\Big[\min_{\bA_j\in\calA_j}D(\hat{Q}_{Y^{n-1}}\|P_j\bA_j)\Big]\leq\gamma_{n-1}\right)\\
&= P_1\bigg( D(\hat{Q}_{Y^{(n-1)}}\|P_1\tilde{\bA}_1)\geq D(\hat{Q}_{Y^{(n-1)}}\|P_1\tilde{\bA}_1)\\
&\qquad-\gamma_{n-1}+\min_{j\neq 1}\Big[\min_{\bA_j\in\calA_j}D(\hat{Q}_{Y^{n-1}}\|P_j\bA_j)\Big]\bigg)\\
&\overset{(a)}{\leq}P_1\bigg(D(\hat{Q}_{Y^{(n-1)}}\|P_1\tilde{\bA}_1)\geq -\gamma_{n-1}+2B^*\bigg)\\
&\overset{(b)}{\leq }\frac{1}{\alpha} (M-1)e^{-(n-1)2B^*} n^{2|\calX|} e^{(n-1)^{\zeta}},
\end{align*}
where $(a)$ is because when $j\neq 1$, $$\min_{\bA_j\in\calA_j}D(\hat{Q}_{Y^{(n-1)}}\|P_j\bA_j)+D(\hat{Q}_{Y^{(n-1)}}\|P_1\tilde{\bA}_1)\geq 2B^*$$ and $(b)$ is from Lemma~\ref{lem:concen}.
\end{proof}

Based on Lemma~\ref{lem:finite}, under $H_i$, for every $0<\alpha\leq 1$, we have
\begin{align}
\label{eqn:stoppingtime}
P_i(T^*=\infty)\leq \lim_{n\to\infty} P_i(T^*\geq n) =0.
\end{align}
This means that the stopping time $T^*_i,i\in[M]$ are almost surely finite when $0<\alpha\leq 1$.
Thus, based on~\eqref{eqn:stop},~\eqref{eqn:decision} and~\eqref{eqn:stoppingtime}, we have that there exists $i\in[M]$, such that
\begin{align}
\min_{ j\in[M],j\neq i}\bigg[\min_{\bA_j\in\calA_j}D(\hat{Q}_{Y^{T^*}}\|P_j\bA_j)\bigg]&\geq \gamma_{T^*} ,\label{eqn:stop1}\\
\min_{ j\in[M],j\neq i}\bigg[\min_{\bA_j\in\calA_j}D(\hat{Q}_{Y^{T^*-1}}\|P_j\bA_j)\bigg]&\leq \gamma_{T^*-1}\label{eqn:stop2}.
\end{align}

Define $\tilde{Q}_i=P_i\tilde{\bA}_i$ for $i \in [M]$. Next, by observing that for any distribution (probability mass function) $Q$, $D(Q\|\tilde{Q}_i)\leq -\log{\min_{y\in\calX} \tilde{Q}_i(y)}$. Denote $Q_{\max}:=\max_{i\in[M]}\{-\log{\min_{y\in\calX} \tilde{Q}_i(y)}\}$, we get from~\eqref{eqn:stop1} that 
\begin{align*}
&P_i(T^*\leq n) \\
&\leq \sum_{j=1}^M P_i\bigg(T^* Z_j^{(T^*)} >\log\frac{1}{\alpha}, T^*\leq n \bigg)\\
&\leq M P_i\bigg(nQ_{\max}>\log\frac{1}{\alpha}, T^*\leq n \bigg)=0, \quad\forall\, n<\frac{\log\frac{1}{\alpha}}{Q_{\max}},
\end{align*}
which yields that $T^*\to\infty$ as $\alpha\to 0^+$,   $P_i$-a.s.  

Since we assumed that the adversary's strategy is $(\tilde{\bA}_1,\ldots, \tilde{\bA}_M)$,   the true distribution of $Y^n$ is $P_i\tilde{\bA}_i$ under $H_i$. Thus, under $H_i$, by the strong law of large numbers, we have that $\hat{Q}_{Y^n}\to P_i\tilde{\bA}_i$ a.s.\ as $n\to\infty$. Consequently, we conclude from the continuity of $D(\cdot\|P_j\bA_j)$ on the finite alphabet $\calX$ that under $H_i$,  $D(\hat{Q}_{Y^{T^*}}\|P_j\bA_j)\to D(P_i\tilde{\bA}_i\|P_j\bA_j)$ a.s.\ as $\alpha\to0^+$ for each $\bA_j\in\calA_j$.   Thus, we have shown the pointwise convergence for each  $\bA_j\in\calA_j$.   Now we prove the uniform almost sure convergence of $\min_{\bA_j\in\calA_j}D(\hat{Q}_{Y^{T^*}}\|P_j\bA_j)$.

Recall that $\calA_j$ is assumed to be a compact set.  Note that $D(\tilde{Q}_i\|Q_j)$ is strongly convex with respect to $Q_j$. Hence
there is a unique $P_j\bA_j$ that minimizes $D(P_i\tilde{\bA}_i\|P_j\bA_j)$. 
We also need to show that  $D(\hat{Q}_{Y^{T^*}}\|P_j\bA_j)$ is {\em stochastically equicontinuous}. That is, for every $\epsilon>0$, there exists a $\delta>0$ such that
\begin{align}
\lim_{\alpha\to 0^+}\!P_i\!\Bigg(\!\sup_{\substack{Q_j, Q_j'\in\calA_j : \\\|Q_j\!-Q_j'\|_1\!\leq\delta}}\!|D( \hat{Q}_{Y^{T^*}}\|Q_j\!)\!-D( \hat{Q}_{Y^{T^*}} \|Q_j' )|>\! \epsilon\!\Bigg)\!=\!0. \label{eqn:stoc_equi}
\end{align}
At this point, we note that for every $\epsilon>0$ and for $0<\delta<\epsilon\min_{Q_j\in\calA_j}\min_{y\in\calX}Q_j(y)$, 
\begin{align*}
&P_i\bigg(\sup_{\|Q_j-Q_j'\|_1\leq\delta}|D(\hat{Q}_{Y^{T^*}}\|Q_j)-D(\hat{Q}_{Y^{T^*}}\|Q_j')|> \epsilon\bigg)\\
&\leq P_i\bigg(\sup_{\|Q_j-Q_j'\|_1<\delta}\sum_{a\in\calX}\hat{Q}_{Y^{T^*}}(a)\bigg|\log\frac{Q_j(a)}{Q'_j(a)}\bigg|>\epsilon\bigg)\\
&\overset{(a)}{\leq} P_i\bigg(\sup_{\|Q_j-Q_j'\|_1<\delta}\sum_{a\in\calX} \hat{Q}_{Y^{T^*}}(a) \frac{|Q_j(a)-Q'_j(a)|}{\min_{y\in\calX} Q_j(y)}>\epsilon\bigg)\\
&\le P_i\bigg(\sup_{\|Q_j-Q_j'\|_1<\delta}\sum_{a\in\calX} \frac{|Q_j(a)-Q'_j(a)|}{\min_{y\in\calX} Q_j(y)}>\epsilon\bigg)\\
&\le P_i\bigg(\sup_{\|Q_j-Q_j'\|_1<\delta} \frac{\|Q_j-Q'_j\|_{1}}{\min_{Q_j\in\calA_j}\min_{y\in\calX} Q_j(y)}>\epsilon\bigg)\\
&\stackrel{(b)}{=}0,
\end{align*}
where $(a)$ is because for any $x,y\geq \beta$, $|\log x-\log y|\leq \frac{1}{\beta}|x-y|$ and $(b)$ follows from the choice of $\delta$. 
Therefore we show that $D(\hat{Q}_{Y^{T^*}}\|P_j\bA_j)$ is stochastically equicontinuous.
Then based on the stochastic Arzel\`a--Ascoli lemma~\cite[Theorem~14.3.2]{rao2008course}, we have
\begin{align}
\lim_{\alpha\to 0^+}\min_{\bA_j\in\calA_j}D(\hat{Q}_{Y^{T^*}}\|P_j\bA_j)&\stackrel{\text{a.s.}}{=} \min_{\bA_j\in\calA_j} D(P_i\tilde{\bA}_i\|P_j\bA_j),\label{eqn:display1}
\end{align}
and
\begin{align}
\lim_{\alpha\to 0^+}\min_{\bA_j\in\calA_j}D(\hat{Q}_{Y^{(T^*-1)}}\| P_j\bA_j)&\stackrel{\text{a.s.}}{=}  \min_{\bA_j\in\calA_j}  D(P_i\tilde{\bA}_i\|P_j\bA_j).\label{eqn:display2}
\end{align}
Combining above results with~\eqref{eqn:stop1} and~\eqref{eqn:stop2}, we deduce that  under hypothesis~$H_i$, 
\begin{align}
\lim_{\alpha\to 0^+}\frac{T^*}{\log(1/\alpha)} \stackrel{\text{a.s.}}{=}  \frac{1}{\min_{j\neq i}\big[\min_{\bA_j\in\calA_j}D(P_i\tilde{\bA}_i\|P_j\bA_j)\big]}. \label{eqn:almostsure}
\end{align}

To go from a.s.\ convergence   above to convergence in mean, it   suffices to prove that there exists an $\epsilon_0>0$ such that a family of random variables $\big\{\frac{T^*}{\log(1/\alpha)} \big\}_{0<\alpha\leq\epsilon_0 }$ is uniformly integrable. That is, there exists an $\epsilon_0>0$ such that for all  $\alpha\in(0,\epsilon_0]$,
\begin{align*}
\lim_{\eta\to\infty}\bbE_i\bigg[\frac{T^*}{\log(1/\alpha)}\mathbbm{1}_{\left\{{T^*}/{\log(1/\alpha)}\geq \eta\right\}}\bigg]=0.
\end{align*}
Here we choose an $\epsilon_0\in (0,1)$ such that $x \in (0,\infty)\mapsto x\log(1/x)$ is increasing on $(0, \epsilon_0]$ and $\frac{\log(1/\epsilon_0)}{1/\epsilon_0}\leq 1$. We now choose $\eta>0$ such that $\eta B^*\geq 2|\calX|+2$. Then for any $0<\alpha\leq\epsilon_0$, we have the derivation shown  in~\eqref{eqn:mid1} (on the top of next page),
\begin{figure*}
\begin{align}
&\bbE_i\bigg[\frac{T^*}{\log(1/\alpha)}\mathbbm{1}_{\left\{{T^*}/{\log(1/\alpha)}\geq \eta\right\}}\bigg]\notag\\
&\leq \bbE_i\bigg[\frac{T^*-\lfloor\eta\log(\frac{1}{\alpha})\rfloor+\eta\log(\frac{1}{\alpha})}{\log(\frac{1}{\alpha})}\mathbbm{1}_{\{{T^*}\geq \lfloor\eta\log(1/\alpha)\rfloor\}}\bigg]\notag\\
& \leq \frac{1}{\log(1/\alpha)} \sum_{l=1}^{\infty} P_i[T^*\geq \lfloor\eta\log({1}/{\alpha})\rfloor+l]+\eta P_i[T^*\geq \lfloor\eta\log({1}/{\alpha})\rfloor]\notag\\
&\leq  \frac{1/\alpha}{\log(1/\alpha)} \sum_{l=1}^{\infty} (\lfloor\eta\log(1/\alpha)\rfloor+l)^{2|\calX|} e^{-(\eta\log(1/\alpha)+l-2)2B^*+(\eta\log(1/\alpha)+l)^{1-\zeta}}\notag\\
&\qquad + \frac{\eta}{\alpha} e^{-(\eta\log(1/\alpha)-2)2B^*+(\eta\log(1/\alpha))^{1-\zeta}}(\lfloor\eta\log(1/\alpha)\rfloor)^{2|\calX|}\notag\\
&\leq \frac{1/\alpha}{\log(1/\alpha)} \sum_{l=1}^{\infty} (\lfloor\eta\log(1/\alpha)\rfloor+l)^{2|\calX|} e^{-(\eta\log(1/\alpha)+l-4)B^*}+ \frac{\eta}{\alpha} e^{-(\eta\log(1/\alpha)-4)B^*}(\lfloor\eta\log(1/\alpha)\rfloor)^{2|\calX|}\notag\\
&\overset{(a)}{\leq} \frac{1/\alpha}{\log(1/\alpha)} \sum_{l=1}^{\infty} 2^{2|\calX|-1}({\lfloor\eta\log(1/\alpha)\rfloor}^{2|\calX|}+l^{2|\calX|} ) e^{-(\eta\log(1/\alpha)+l-4)B^*}+\frac{ \eta}{\alpha} e^{-(\eta\log(1/\alpha)-4)B^*}(\lfloor\eta\log(1/\alpha)\rfloor)^{2|\calX|}\notag\\
&=\frac{1/\alpha}{\log(1/\alpha)}2^{2|\calX|-1}e^{4B^*}e^{-\eta\log(1/\alpha)B^*}\sum_{l=1}^{\infty}l^{2|\calX|}e^{-lB^*}+\frac{\eta}{\alpha} e^{-\eta\log(1/\alpha)B^*}e^{4B^*}(\lfloor\eta\log(1/\alpha)\rfloor)^{2|\calX|}\notag\\
&\qquad+\frac{1/\alpha}{\log(1/\alpha)} 2^{2|\calX|-1}{\lfloor\eta\log(1/\alpha)\rfloor}^{2|\calX|}e^{4B^*}e^{-\eta\log(1/\alpha)B^*}\sum_{l=1}^{\infty}e^{-lB^*}\notag\\
&\leq C_1\bigg(\frac{1}{\alpha}\bigg)^{1-\eta B^*}\!\frac{1}{\log (1/\alpha)}+ C_2\eta^{2|\calX|+1}\!\bigg(\!\frac{1}{\alpha}\!\bigg)^{\!1-\!\eta B^*}\!\bigg(\!\log\frac{1}{\alpha}\!\bigg)^{2|\calX|}+ C_3 \eta^{2|\calX|}\bigg(\log\frac{1}{\alpha}\bigg)^{2|\calX|-1}\bigg(\frac{1}{\alpha}\bigg)^{1-\eta B^*}\notag\\
&\stackrel{(b)}{\le} C_1\frac{\epsilon_0^{\eta B^*-1}}{\log (1/\epsilon_0)}+ C_2\eta^{2|\calX|+1}\epsilon_0^{\eta B^*-1-2|\calX|}\!\bigg(\!\epsilon_0\log\frac{1}{\epsilon_0}\!\bigg)^{2|\calX|}+ C_3 \eta^{2|\calX|}\bigg(\epsilon_0\log\frac{1}{\epsilon_0}\bigg)^{2|\calX|-1}\epsilon_0^{\eta B^*-2|\calX|}.\label{eqn:mid1}
\end{align}\hrulefill
\end{figure*}
where  $C_1:=2^{2|\calX|-1}e^{4B^*}\sum_{l=1}^{\infty}l^{2|\calX|}e^{-lB^*}$, $C_2:=e^{4B^*}$, and $C_3:=2^{2|\calX|-1}e^{4B^*}\sum_{l=1}^{\infty}e^{-lB^*}$. 
In the derivation of~\eqref{eqn:mid1}, $(a)$ follows from the inequality $(x+y)^k\le 2^{k-1}(x^k+y^k)$ for any $x,y>0$ and any integer $k$, and $(b)$ follows from the choice of $\epsilon_0$. As~\eqref{eqn:mid1} tends to $0$ as $\eta\to \infty$, we have proved the uniform integrability of the family of random variables $\big\{\frac{T^*}{\log(1/\alpha)}\big\}_{0<\alpha\leq\epsilon_0}$. Thus, under $H_i$, we have 
\begin{equation}
\label{eqn:stoptime1}
\lim_{\alpha\to 0^+} \frac{\bbE_i[T^*]}{\log(1/\alpha)}=\frac{1}{\min_{ j\neq i}\big[\min_{\bA_j\in\calA_j}D(P_i\tilde{\bA}_i\|P_j\bA_j)\big]}.
\end{equation}
Therefore, when the adversary's strategy  is $(\tilde{\bA}_1,\ldots,\tilde{\bA}_M)$, the asymptotic payoff function for the test $\Phi^{*}$ as $\alpha\to 0^+$ is
\begin{align}
\label{eqn:optpayoff}
&\lim_{\alpha\to 0^+} u_{\blambda}^{(\alpha)}(\Phi^*,(\tilde{\bA}_1,\ldots,\tilde{\bA}_M))\notag\\
&=\sum_{i=1}^M\lambda_i\min_{ j\in[M]\setminus\{ i \}}\Big[\min_{\bA_j\in\calA_j} D(P_i\tilde{\bA}_i\|P_j\bA_j)\Big].
\end{align} 
\emph{Part 3: Proof of the asymptotic Nash equilibrium at the profile $(\Phi^*,(\bA_1^*,\ldots,\bA_M^*))$.}
We first prove~\eqref{eqn:opt1} by deriving a lower bound on the expected sample size for a general sequential test. For a fixed pair $(\tilde{\bA}_1,\ldots,\tilde{\bA}_M)$, we consider the multiple hypothesis testing problem  in which under $H_i$ the observations are generated from $P_i\tilde{\bA}_i,i\in[M]$. Let $\delta=(d,T)$ be a sequential test for the multiple hypothesis testing problem. Let $\alpha_{ij}$ be the probability that $H_j$ is accepted when $H_i$ is the underlying hypothesis. Let $\alpha_i$ be the probability that $H_i$ is rejected when $H_i$ is the true hypothesis. Note that $\alpha_i=1-\alpha_{ii}=\sum_{j\not=i}\alpha_{ij}$. We choose $\alpha\in(0,\frac{1}{M})$ as the upper bound for the error probabilities $\alpha_i, i\in[M]$.  Based on~\cite[Lemma~4.3.1]{tartakovsky2014sequential}, Then we have that for all $i\in [M]$,
\begin{align}
\bbE_i[\tilde{T}]&\stackrel{(a)}{\ge} \max_{j\in[M],j\neq i}\frac{1}{D(P_i\tilde{\bA}_i\|P_j\tilde{\bA}_j)}\notag\\
&\qquad \times\bigg(\sum_{k\neq i}\alpha_{ik}\log\frac{\alpha_{ik}}{\alpha_{jk}}+(1-\alpha_i)\log\frac{1-\alpha_i}{\alpha_{ij}}\bigg)\notag\\
&\stackrel{(b)}{\geq} \max_{j\in[M],j\neq i}\frac{1}{D(P_i\tilde{\bA}_i\|P_j\tilde{\bA}_j)}\\
&\qquad \times\bigg(\sum_{k\neq i}\alpha_{ik}\log\frac{\sum_{k\neq i}\alpha_{ik}}{\sum_{k\neq i}\alpha_{jk}}+(1-\alpha_i)\log\frac{1-\alpha_i}{\alpha_{ij}}\bigg)\notag\\
&\geq \max_{j\in[M],j\neq i}\frac{1}{D(P_i\tilde{\bA}_i\|P_j\tilde{\bA}_j)}\notag\\
&\qquad \times\bigg(\alpha_i\log\frac{\alpha_i}{1-\alpha_i}+(1-\alpha_i)\log\frac{1-\alpha_i}{\alpha_{ij}}\bigg)\notag\\
&\stackrel{(c)}{\ge} \max_{j\in[M],j\neq i}\frac{1}{D(P_i\tilde{\bA}_i\|P_j\tilde{\bA}_j)}(1-\alpha)\log\frac{1-\alpha}{\alpha},\label{eqn:DPI}
\end{align}
where $(a)$ follows from~\cite[Lemma~4.3.1]{tartakovsky2014sequential}, $(b)$ follows from the log sum inequality~\cite[Theorem 2.6.1]{cover2006elements} and $(c)$ follows from the fact that $\alpha_i\in(0,1/M)$.

We denote $\Phi'=(T',\delta')$ as an arbitrary sequential test to test the $M$-ary composite hypothesis testing problem $\{H_i:P_i\bA_i\}_{i=1}^M,\bA_i\in\calA_i$ with their error probabilities upper bounded by $\alpha$.
Then we use the test $\Phi'$ for the $M$-ary hypothesis testing problem $\{H_i:P_i\tilde{\bA}_i\}_{i=1}^M$ for any fixed $(\tilde{\bA}_1,\ldots,\tilde{\bA}_M)$. Their error probabilities are also upper bounded by $\alpha$. From~\eqref{eqn:DPI} we obtain that for all $i\in[M]$,
\begin{align}
\bbE_i[T']&\geq\max_{j\neq i}\frac{(1-\alpha)\log\frac{1-\alpha}{\alpha}}{D(P_i\tilde{\bA}_i\|P_j\tilde{\bA}_j)}.  \label{eqn:ineq1}
\end{align}
In  \eqref{eqn:ineq1}, the left-hand side $\bbE_i[T']$ does not depend on $\{\tilde{\bA}_j:j\not=i\}$. 
 We  recall that $\calA_i$ is a compact set and $\calX$ is finite and $Q_i$ and $Q_j$ have full support on $\calX$.  So we can maximize  the right-hand side with respect to $\bA_j$ and obtain that 
\begin{align}
\label{eqn:compositestop}
\bbE_i[T']\geq \max_{j\neq i}\max_{\bA_j\in\calA_j}\frac{(1-\alpha)\log\frac{1-\alpha}{\alpha}}{D(P_i\tilde{\bA}_i\|P_j\bA_j)}.
\end{align}
Thus, based on~\eqref{eqn:compositestop}, we have
\begin{align*}
\lim_{\alpha\to 0^+} \frac{\bbE_i[T']}{\log(1/\alpha)}\geq \frac{1}{\min_{j\neq i}\Big[\min_{\bA_j\in\calA_j}D(P_i\tilde{\bA}_i\|P_j\bA_j)\Big]}.
\end{align*}
As~\eqref{eqn:compositestop} holds for any test $\Phi\in\calS_{\rmD}(\alpha)$, we have
\begin{align*}
&\lim_{\alpha\to 0^+}\sup_{\Phi\in\calS_{\rmD}(\alpha)}u_{\lambda}^{(\alpha)}(\Phi,(\bA_1^*,\ldots,\bA_M^*))\\
&\leq \sum_{i=1}^M\lambda_i\min_{ j\in[M],j\neq i}\Big[\min_{\bA_j\in\calA_j} D(P_i{\bA}^*_i\|P_j\bA_j)\Big]\\
&=\lim_{\alpha\to 0^+}u_{\lambda}^{(\alpha)}(\Phi^{*},(\bA_1^*,\ldots,\bA_M^*)),
\end{align*}
which completes the proof of~\eqref{eqn:opt1}.

To complete the proof of~\eqref{eqn:opt2}, we leverage a  technical lemma.
\begin{lemma}\label{lem:uniform}
For all $i\in[M]$, it holds that the family of   numbers
	$\big\{\frac{\bbE_i[T^*]}{\log(1/\alpha)}\big\}_{0<\alpha\leq 1}$ converges uniformly on $\calA_i$ as $\alpha\to 0^+$. 
\end{lemma} 
The proof of Lemma~\ref{lem:uniform} is presented in   Appendix~\ref{sec:lemma5}. From~\eqref{eqn:optpayoff}, we have that
\begin{align*}
&\lim_{\alpha\to 0^+}-u_{\blambda}^{(\alpha)}(\Phi^*,(\tilde{\bA}_1,\ldots,\tilde{\bA}_M))\\
&=-\sum_{i=1}^M\lambda_i\min_{ j\in[M]\setminus\{ i \}}\Big[\min_{\bA_j\in\calA_j} D(P_i\tilde{\bA}_i\|P_j\bA_j)\Big].
\end{align*}
It is obvious from the definition of $\bA_i^{*}, i\in[M]$ in~\eqref{eqn:opta0} that
\begin{align}
&\lim_{\alpha\to 0^+}-u_{\lambda}^{(\alpha)}(\Phi^{*},(\bA_1^*,\ldots,\bA_M^*))\notag\\
&=-\sum_{i=1}^M\lambda_i\min_{j\neq i}\Big[\min_{\bA_j\in\calA_j} D(P_i{\bA}^*_i\|P_j\bA_j)\Big]\notag\\
&\geq \sup_{(\tilde{\bA}_1,\ldots,\tilde{\bA}_M)\in\calS_{\rmA}(\Delta)}-\sum_{i=1}^M\lambda_i\min_{j\neq i}\Big[\min_{\bA_j\in\calA_j} D(P_i\tilde{\bA}_i\|P_j\bA_j)\Big]\notag\\
&=\lim_{\alpha\to 0^+}\sup_{(\tilde{\bA}_1,\ldots,\tilde{\bA}_M)\in\calS_{\rmA}(\Delta)}-u_{\lambda}^{(\alpha)}(\Phi^{*},(\tilde{\bA}_1,\ldots,\tilde{\bA}_M))\label{eqn:uniform}, 
\end{align} 
where~\eqref{eqn:uniform} is based on Lemma~\ref{lem:uniform} (see   Appendix~\ref{sec:lemma5} for details). Thus, the proof of~\eqref{eqn:opt2} is complete.

The proof of 
Theorem~\ref{thm:advboth} is completed by combining the above three parts.
\end{proof}

\section{Extension to the Adversary Non-awareness Setting}
\label{sec:nonawareness}

In this section, we consider the case when the adversary {\em does not} know the underlying distribution of observed samples in a {\em binary} hypothesis test. In this case, referring to Fig.~\ref{fig:actboth}, the adversary can only apply a   {\em common} perturbation mechanism $\bA$ to the two hypotheses. 

We define the expectation of the stopping time under $H_i$ as $\bbE_i[\tau]$ for $i\in\{0,1\}$ and a (non-aware) test is a pair $\Phi_{\mathrm{NA}}=(\tau,\delta_{\mathrm{NA}})$. We also define the adversary's and decision maker's strategy sets as
\begin{align*}
    \hat{\calS}_{\rmA}(\Delta):=&\bigg\{\bA:\max_{i\in\{0,1\}}d(P_i,P_i\bA)\le \Delta\bigg\},\quad\mbox{and}\\
    \hat{\calS}_{\rmD}(\alpha):=&\bigg\{\Phi_{\mathrm{NA}}:\max_{i\in\{0,1\}}\sup_{\bA\in\calS_{\rmA}(\Delta)}\alpha_i(\Phi_{\mathrm{NA}})\leq\alpha\bigg\},
\end{align*}
respectively, and the payoff function of the decision maker as
\begin{align*}
    \hat{u}_{\lambda}^{(\alpha)}(\Phi_{\mathrm{NA}},\bA):= \frac{\log(1/\alpha)}{\bbE_0[\tau]}+\lambda\frac{\log(1/\alpha)}{\bbE_1[\tau]}.
\end{align*}
We define 
\begin{align*}
    S_n:=\min_{\bA\in\hat{\calS}_{\rmA}(\Delta)}\max\{D(\hat{Q}_{Y^n}\|P_0\bA),D(\hat{Q}_{Y^n}\|P_1\bA)\}.
\end{align*}
Then we define the stopping time as
\begin{align*}
    \tau^*:=\inf\big\{n\geq 1: S_n\geq\gamma_n\big\},
\end{align*}
and the decision rule as 
\begin{align*}
    \delta_{\mathrm{NA}}^* (Y^n):=
    \begin{cases}
    0,& \mbox{if~} \min_{\bA\in\hat{\calS}_{\rmA}(\Delta)}D(\hat{Q}_{Y^n}\|P_1\bA)\geq\gamma_n,\\
    1,& \mbox{if~}\min_{\bA\in\hat{\calS}_{\rmA}(\Delta)}D(\hat{Q}_{Y^n}\|P_0\bA)\geq\gamma_n.
    \end{cases}
\end{align*}
Then our adversary non-aware test is $\Phi_{\mathrm{NA}}^*=(\tau^*,\delta_{\mathrm{NA}}^*)$. We obtain the following two propositions which present the achievable and converse results respectively. 
\begin{proposition}
\label{prop:lower}
If $\calS_{\rmA}(\Delta)$ is a compact set, then for any $\lambda>0$, we have that
\begin{align}
\label{eqn:lower}
    &\lim_{\alpha\to 0^+}\hat{u}_{\lambda}^{(\alpha)}(\Phi^*_{\mathrm{NA}},\tilde{\bA})\notag\\
    &= \min_{\bA\in\hat{\calS}_{\rmA}(\Delta)}\max\Big\{ D(P_0\tilde{\bA}\|P_1\bA),D(P_0\tilde{\bA}\|P_0\bA)\Big\}\notag\\
    &\;\;+\!\lambda\! \min_{\bA\in\hat{\calS}_{\rmA}(\Delta)}\!\max\!\Big\{\! D(P_1\tilde{\bA}\|P_0\bA),\!D(P_1\tilde{\bA}\|P_1\bA)\!\Big\}.\!
\end{align}
\end{proposition}
\begin{proof}[Proof sketch of Proposition~\ref{prop:lower}]
 Based on the definition of $\delta_{\mathrm{NA}}^*$, following the same procedure as in the first part of proof in Section~\ref{sec:main}, we can  prove that $\Phi_{\mathrm{NA}}^*\in\hat{\calS}_{\rmD}(\Delta)$. To obtain the error exponents for the test $\Phi_{\mathrm{NA}}^*$ and any $\tilde{\bA}\in\hat{\calS}_{\rmA}(\Delta)$, we first prove a result similar to  Lemma~\ref{lem:finite}. 
 Set $M=2$ in $\gamma_{n-1}$. Then we have
 \begin{align*}
     P_0(\tau^*\geq n)&\leq P_0(S_{n-1}\leq \gamma_{n-1})\\
     &\leq P_0\bigg(\min_{\bA\in\hat{\calS}_{\rmA}(\Delta)}D(\hat{Q}_{Y^{n-1}}\|P_0\bA)\leq \gamma_n\bigg)\\
     &\leq P_0\bigg(D(\hat{Q}_{Y^{n-1}}\|P_1\tilde{\bA})\geq -\gamma_n+2\hat{B}^{*}\bigg)\\
     &\leq \frac{1}{\alpha} e^{-(n-1)2\hat{B}^{*}} n^{2|\calX|} e^{(n-1)^{\zeta}},
 \end{align*}
 where $\hat{B}^{*}:=\min_{\bA\in\hat{\calS}_{\rmA}(\Delta)}B(P_0\bA,P_1\bA)$. Then following the same procedure as in the second part of proof in Section~\ref{sec:main} but now the limit is
\begin{align*}
        \lim_{\alpha\to 0^+}\!\frac{\log{\frac{1}{\alpha}}}{\bbE_0[\tau^*]}\!=\! \min_{\bA\in\hat{\calS}_{\rmA}(\Delta)}\!\max\!\Big\{ \!D(P_0\tilde{\bA}\|P_1\bA),D(P_0\tilde{\bA}\|P_0\bA)
        \!\Big\}.
\end{align*}
We can obtain an analogous  result for $\lim_{\alpha\to 0^+}\frac{\log{\frac{1}{\alpha}}}{\bbE_1[\tau^*]}$. So combining the above two results, we can obtain~\eqref{eqn:lower}.
\end{proof}
Based on Proposition~\ref{prop:lower}, we see that the adversary can choose the strategy that minimizes the achievable bound of the decision maker:
\begin{align*}
    \bA^*=&\argmin_{\tilde{\bA}\in\hat{\calS}_{\rmA}(\Delta)}\!\bigg[\!\min_{\bA\in\hat{\calS}_{\rmA}(\Delta)}\!\max\!\Big\{\! D(P_0\tilde{\bA}\|P_1\bA), D(P_0\tilde{\bA}\|P_0\bA)\!\Big\}\\
    &\quad+\!\lambda \min_{\bA\in\hat{\calS}_{\rmA}(\Delta)}\!\max\Big\{ D(P_1\tilde{\bA}\|P_0\bA),D(P_1\tilde{\bA}\|P_1\bA)\Big\}\bigg].
\end{align*}
Using this   strategy, we   find that 
\begin{align*}
    \lim_{\alpha\to 0^+}\hat{u}_{\lambda}^{(\alpha)}(\Phi^*_{\mathrm{NA}},\bA^*)\geq \lim_{\alpha\to 0^+}  \hat{u}_{\lambda}^{(\alpha)}(\Phi^*,(\bA_0^*,\bA_1^*)),
\end{align*}
which means that the decision maker can obtain a better (no worse, to be precise) performance than the adversary awareness case. In Proposition~\ref{prop:upper}, we prove a converse bound for any pair of strategies.
\begin{proposition}
\label{prop:upper}
For any test $\Phi_{\mathrm{NA}}\in\hat{\calS}_{\rmD}(\alpha)$ and any $\tilde{\bA}\in\hat{\calS}_{\rmA}(\Delta)$, we have
\begin{equation}
\label{eqn:upper}
   \lim_{\alpha\to 0^+}\!\hat{u}_{\lambda}^{(\alpha)}\!(\Phi_{\mathrm{NA}},\tilde{\bA})\leq D(P_0\tilde{\bA}\|P_1\tilde{\bA})+\lambda D(P_1\tilde{\bA}\|P_0\tilde{\bA}).
\end{equation}
\end{proposition}
\begin{proof}
As the adversary adopts the same strategy on both hypotheses, the problem is equivalent to the following hypothesis testing problem:
$H_0: P_0\tilde{\bA},$ v.s.\ $H_1: P_1\tilde{\bA}$. 
However, for this problem the decision maker have no knowledge of $\tilde{\bA}$. According to the optimality of sequential probability ratio test (SPRT)~\cite{tartakovsky2014sequential}, the upper bound on the error exponents that the decision maker can obtain for any test is
\begin{align*}
    \lim_{\alpha\to 0^+} \hat{u}_{\lambda}^{(\alpha)} (\Phi_{\mathrm{NA}},\tilde{\bA})\le D(P_0\tilde{\bA}\|P_1\tilde{\bA})+\lambda D(P_1\tilde{\bA}\|P_0\tilde{\bA}),
\end{align*}
as desired. 
\end{proof}
\begin{remark}
Observe that the achievable and converse bounds in~\eqref{eqn:lower} and~\eqref{eqn:upper} respectively do not match  in the adversary non-awareness setting. Thus, the pair of strategies $(\Phi_{\mathrm{NA}}^*,\bA^*)$ cannot, in general, achieve the asymptotic Nash equilibrium. However, comparing~\eqref{eqn:lower} to the adversary aware case in~\eqref{eqn:aware_case}, we see that the decision maker  can attain larger error exponents, which implies that the decision maker can perform better. This is aligned with our intuition, since now the adversary is weaker as it has to use the same $\bA$ under both hypotheses. 
\end{remark}

\section{Numerical Experiments}
\label{sec:example}
In this section, we provide two sets of experiments
to corroborate the theory developed in the previous sections. The first uses   synthetic data on a binary hypothesis testing problem with Bernoulli distributions to show that empirical stopping time converges to its theoretical counterpart. The second set of experiments shows that empirical stopping time converges to its theoretical counterpart on the MNIST dataset.

\subsection{Binary Test for Bernoulli Distributions}
Let the distributions be $\mathrm{Bern}(p_0)$ under $H_0$ and $\mathrm{Bern}(\frac{1}{2})$ under $H_1$, respectively. Without loss of generality, we assume $0<p_0<\frac{1}{2}$. We set the distortion measure $d$ to be the total variation distance and the distortion level to be $\Delta$. For the Bernoulli distribution, the adversary's strategy takes the form
\begin{align*}
    \bA_i=
    \begin{bmatrix}
    a_i  &1-a_i\\
     1-b_i  &b_i
    \end{bmatrix} ,\quad i\in\{0,1\}.
\end{align*}
Using the distortion constraints, we can obtain the relationship between $a_i$ and $b_i$ for $i\in\{0,1\}$ as follows:
\begin{align*}
    |1-2p_0-b_0+p_0(a_0+b_0)|&\leq \frac{\Delta}{2},\\
    |a_1-b_1|&\leq \Delta.
\end{align*}
 Then based on Theorem~\ref{thm:advboth}, we can calculate the optimal adversary's strategy by solving the optimization problem in~\eqref{eqn:opta0}. In the Bernoulli case, the perturbed distributions by the optimal adversary strategy are attained on the boundary (shown in Fig~\ref{fig:bern} with red crosses), which means that 
 \begin{align*}
    1-2p_0-b^*_0+p_0(a^*_0+b^*_0)&= \frac{\Delta}{2},\\
    b^*_1-a^*_1&= \Delta.
 \end{align*}
 The payoff function at the asymptotic Nash equilibrium is
\begin{align*}
    \lambda_1D_{\rmb}\Big(0.5- \frac{\Delta}{2}\Big\|p_0+\frac{\Delta}{2}\Big)+\lambda_2 D_{\rmb}\Big(p_0+\frac{\Delta}{2}\Big\|0.5-\frac{\Delta}{2}\Big),
\end{align*}
where $D_{\rmb}(a\|b):=a\log(\frac{a}{b}) + (1-a)\log(\frac{1-a}{1-b})$ is the binary KL divergence between two Bernoulli distributions with parameters $a,b\in (0,1)$.

\begin{figure}
     \centering
     \includegraphics{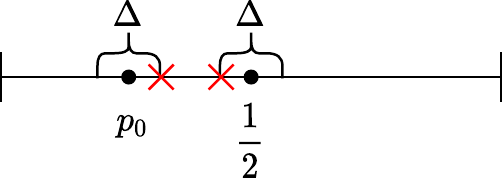}
     \caption{Binary test for Bernoulli distribution. The left red cross is the point whose value equal to $p_0a_0^*+(1-p_0)(1-b_0^*)$, i.e., the perturbed distributions by the optimal strategy $\bA_0^*$. The right red cross is the point with value $0.5 a_1^*+0.5(1-b_1^*)$, i.e., the perturbed distributions by the optimal strategy $\bA_1^*$.}
     \label{fig:bern}
 \end{figure}
Now we set $p_0=0.38$ and $\Delta=0.05$. We can calculate $\bA_i^*, i\in\{0,1\}$ numerically. Note that the optimizing matrices   are not unique. One of the optimizing pairs is
\begin{align*}
    \bA_0^*=\left[
    \begin{matrix}
    0.5  &0.5\\
    0.3419  &0.6581
    \end{matrix}\right]\quad\mbox{and}\quad
    \bA_1^*=\left[
    \begin{matrix}
    0.15  &0.85\\
    0.8  &0.2
    \end{matrix}\right].
\end{align*}
Then in this case, the payoff function at the asymptotic Nash equilibrium is
$
0.0109    \lambda_1 +0.0108\lambda_2 .
$
To corroborate Theorem~\ref{thm:advboth}, now we simulate the sequential adversarial hypothesis testing procedure.

We set $\alpha$ to different values and run the strategy defined in \eqref{eqn:stop}--\eqref{eqn:decision} a total of  $ 50,000$ times for each $\alpha$ to observe the stopping times and hence, the convergence of the payoff function as $\alpha\to 0^+$ under $H_0$ and $H_1$, respectively. The results are shown in Figs.~\ref{fig:limit1} and~\ref{fig:limit2}, respectively. The horizontal line is the theoretical payoff function at the Nash equilibrium, i.e., $\min_{(\bA_0,\bA_1)\in\calS_{\rmA}(\Delta)}D(P_
0\bA_0\|P_1\bA_1)$ or $\min_{(\bA_0,\bA_1)\in\calS_{\rmA}(\Delta)}D(P_1\bA_1\|P_0\bA_0)$. The dotted line is the estimated payoff function, i.e., $\frac{\log(1/\alpha)}{\bbE_{i}[T^*]}$ for $i=0,1$. We observe that as $\alpha\to 0^+$, i.e., $\log(1/\alpha)\to \infty$, $\frac{\log(1/\alpha)}{\bbE_{0}[T^*]}$ converges  to $D_{\rmb}(p_0+\Delta/2\|0.5-\Delta/2)$ under $H_0$ and $\frac{\log(1/\alpha)}{\bbE_{1}[T^*]}$ tends to  $D_{\rmb}(0.5-\Delta/2\|p_0+\Delta/2)$ under $H_1$. 

\begin{figure}
    \centering
    \includegraphics[width=\linewidth]{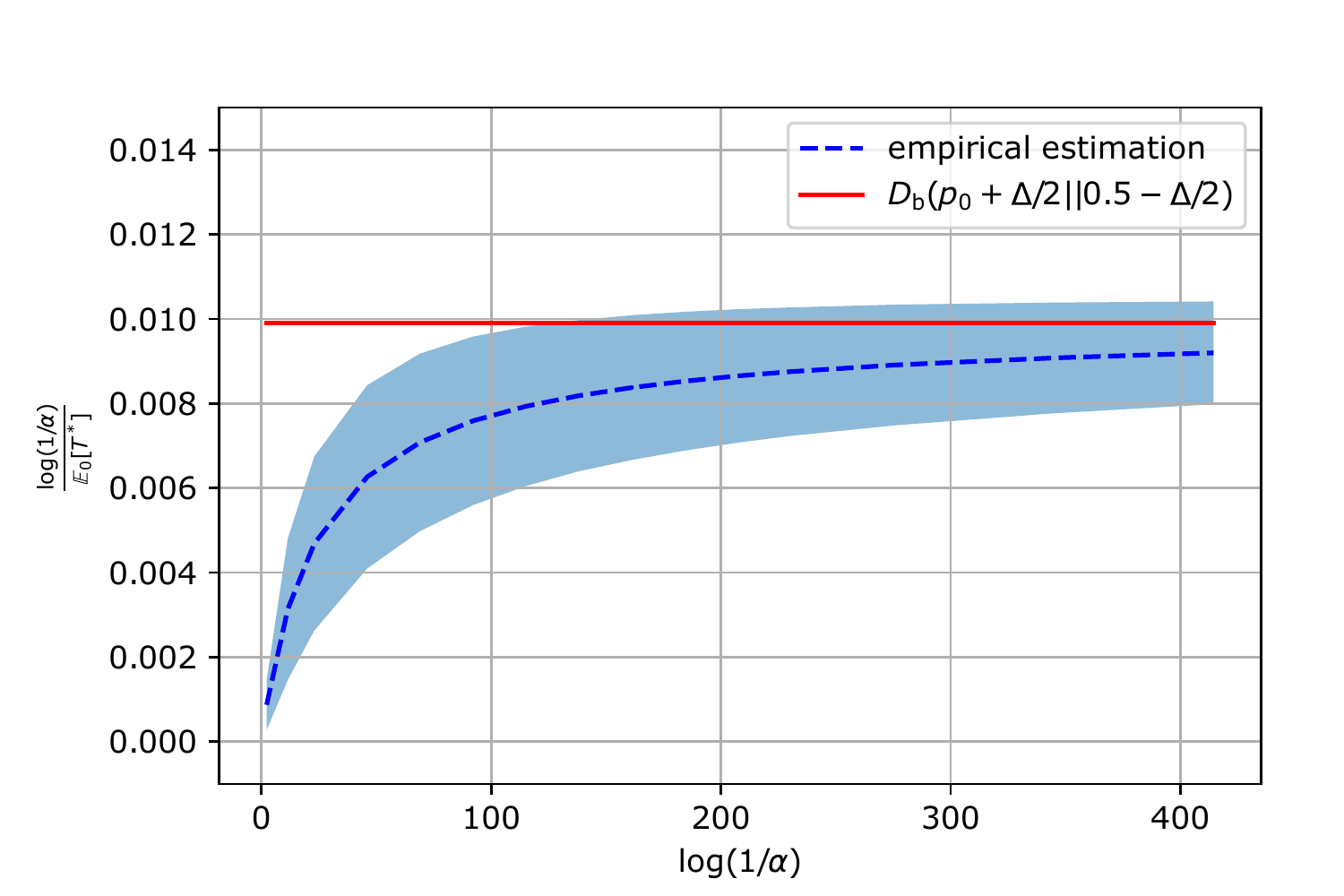}
    \caption{The change of payoff function as $\alpha\to 0^+$ under $H_0$. The shaded part denotes 1 standard deviation above and below the mean across 50,000 independent runs of the strategy defined in \eqref{eqn:stop}--\eqref{eqn:decision}.}
    \label{fig:limit1}
\end{figure}

\begin{figure}
    \centering
    \includegraphics[width=\linewidth]{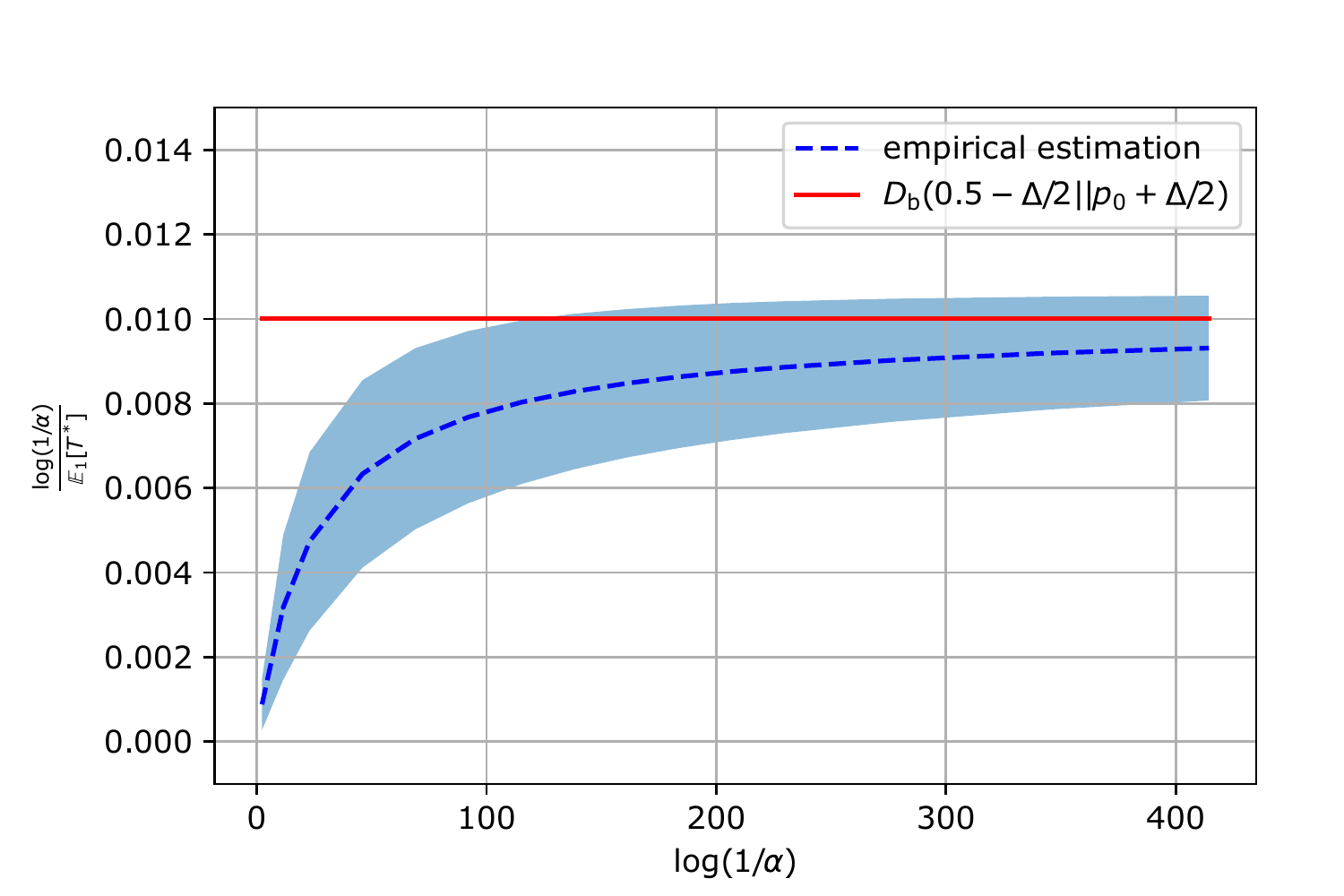}
    \caption{The change of payoff function as $\alpha\to 0^+$  under $H_1$}
    \label{fig:limit2}
\end{figure}

\subsection{Binary Test for the MNIST dataset}
In the previous section, we applied our sequential hypothesis testing strategy to synthetic data. To demonstrate the utility of our strategy on real-world data,  we now apply it to the MNIST dataset. For simplicity, we choose to test two classes from the MNIST dataset---digits 1 and  4. We  also binarize the MNIST data by choose a threshold (here we choose the threshold to be  $50$). When the pixel value greater than the threshold, we set the value to be $255$ and otherwise, we set the value to be $0$. Fig.~\ref{fig:mnist} shows a representative original image and its   binarized version.
\begin{figure}
     \centering
     \begin{subfigure}
         \centering
         \includegraphics[width=0.48\linewidth]{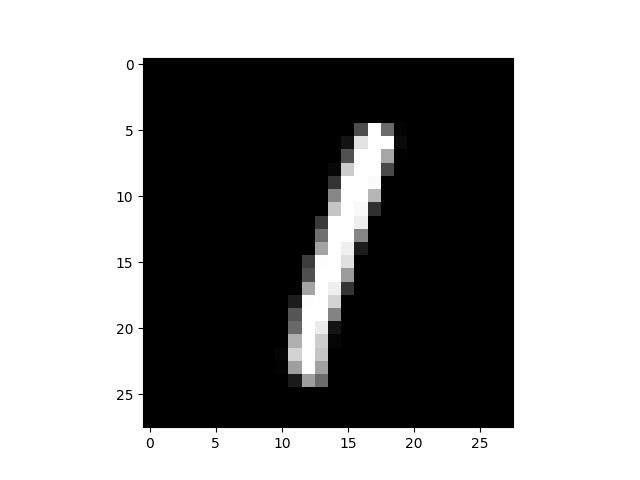}
     \end{subfigure}
     \begin{subfigure}
         \centering
         \includegraphics[width=0.48\linewidth]{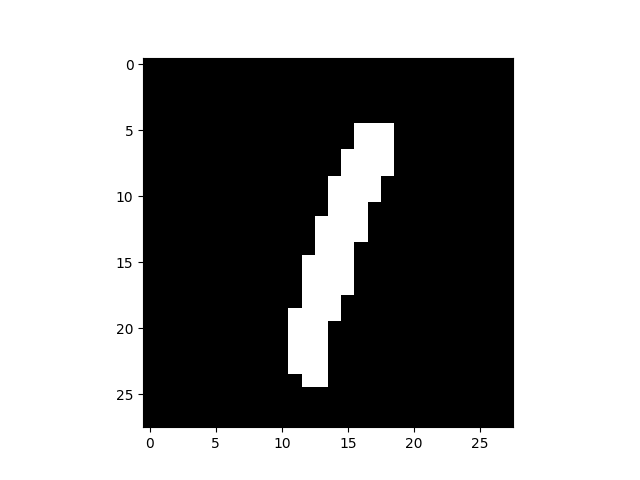}
     \end{subfigure}
     \hfill
    \begin{subfigure}
         \centering
         \includegraphics[width=0.48\linewidth]{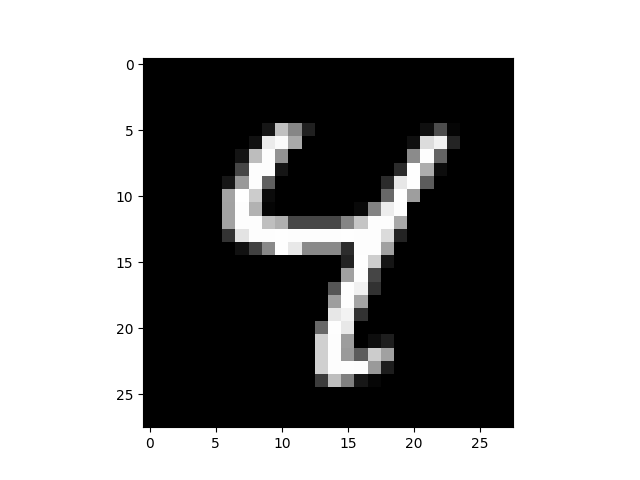}
     \end{subfigure}
    \begin{subfigure}
         \centering
         \includegraphics[width=0.48\linewidth]{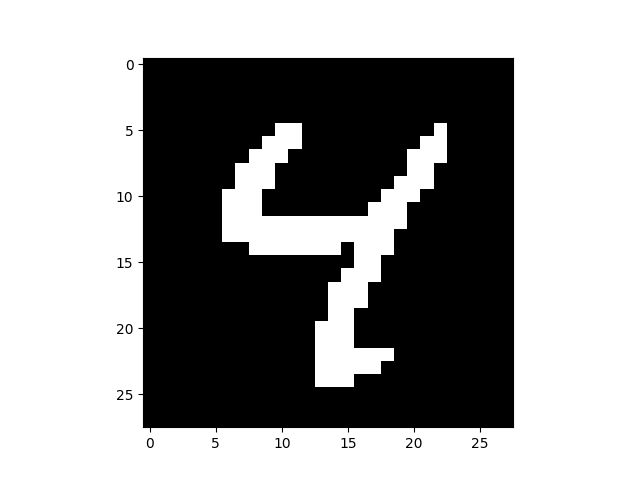}
     \end{subfigure}
    \caption{The images on the left are the the original MNIST ones and those on the right are their binarized counterparts.}
    \label{fig:mnist}
\end{figure}

 Denote the respective distributions for digits $1$ and $4$ respectively as $I_1=[f_{1,0},f_{1,255}]$ and $I_4=[f_{4,0},f_{4,255}]$. We use the training dataset to obtain an empirical estimates of their distributions. We find that  
 \begin{align*}
     I_1=[0.9061,0.09395]\quad\mbox{and}\quad I_4=[0.8481,0.1519].
 \end{align*}
We set the distortion measure to be the KL distance to ensure that $\calS_{\rmA}$ as a convex set and we also set $\Delta=0.001$. We can obtain the adversary's optimal strategy by solving~\eqref{eqn:opta0} numerically using MATLAB's convex optimization toolbox. This yields 
\begin{align*}
    \bA_1^* &=\left[
    \begin{matrix}
    0.94&0.06\\
    0.461 &0.539
    \end{matrix}\right]\quad\mbox{and}\quad
    \bA_4^* =\left[
    \begin{matrix}
    0.9&0.1\\
    0.645&0.355
    \end{matrix}\right].
\end{align*}
Fig.~\ref{fig:perturb} shows two examples of perturbed images produced the adversary. Note that as we use the pixel values to estimate the distribution, the adversary perturbs only the fraction of white and black pixels. 
\begin{figure}
     \centering
     \begin{subfigure}
         \centering
         \includegraphics[width=0.48\linewidth]{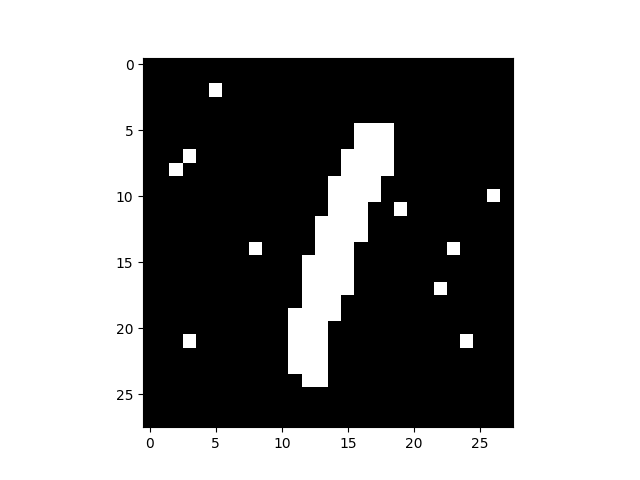}
     \end{subfigure}
     \begin{subfigure}
         \centering
         \includegraphics[width=0.48\linewidth]{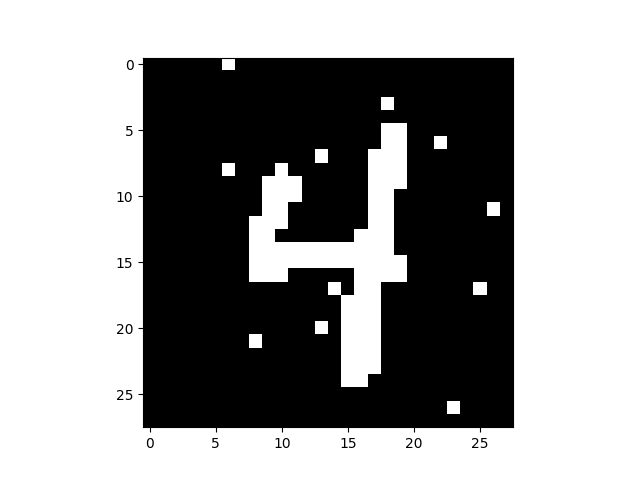}
     \end{subfigure}
    \caption{Two examples of  images perturbed by the adversary.}
    \label{fig:perturb}
\end{figure}

Now we use the test dataset to perform the sequential test. Due to the limited number of samples in the training and test data sets (there are $7877$  images for digit $1$ and $6824$   images for digit $4$), during the testing process, we use a resampling procedure to obtain more test images.
Fig.~\ref{fig:mnist4} and Fig.~\ref{fig:mnist1} show the change of payoff function as $\alpha\to 0^+$ when the true digit is $4$ and the true digit is $1$, respectively. We observe that when $\alpha\to 0^+$, $\frac{\log(1/\alpha)}{\bbE_{1}[T^*]}\to D(I_1\bA^*_1\|I_4\bA^*_4)$ when the true digit is~$1$ and $\frac{\log(1/\alpha)}{\bbE_{4}[T^*]}\to D(I_4\bA^*_4\|I_1\bA^*_1)$ when the true digit is~$4$. Thus the conclusion here is the same as that for     synthetic data, i.e., the promised fundamental limit is attained as $\alpha\to 0^+$.
\begin{figure}[t]
    \centering
    \includegraphics[width=\linewidth]{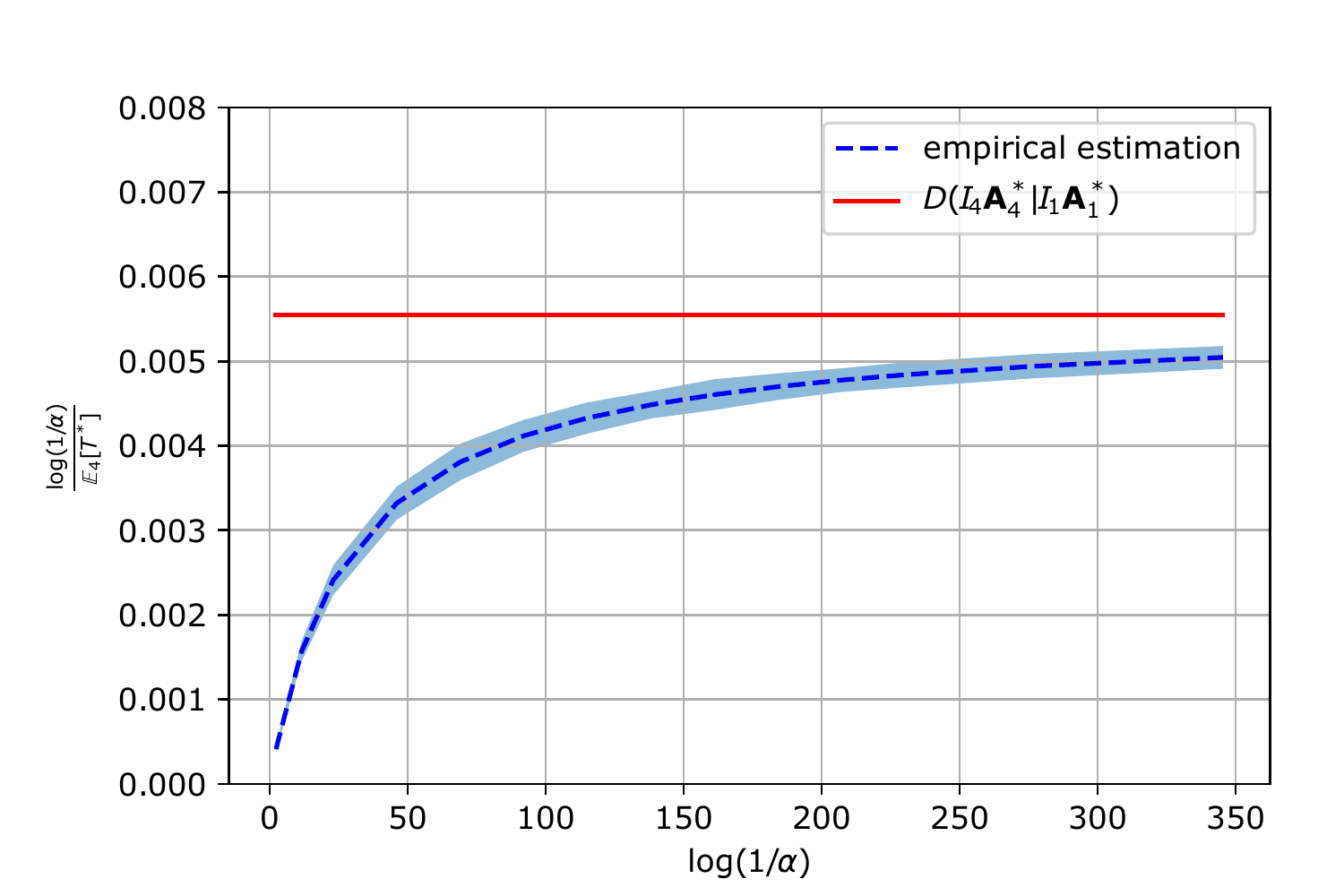}
    \caption{The change of the payoff function as $\alpha\to 0^+$  for digit $4$.}
    \label{fig:mnist4}
\end{figure}

\begin{figure}[t]
    \centering
    \includegraphics[width=\linewidth]{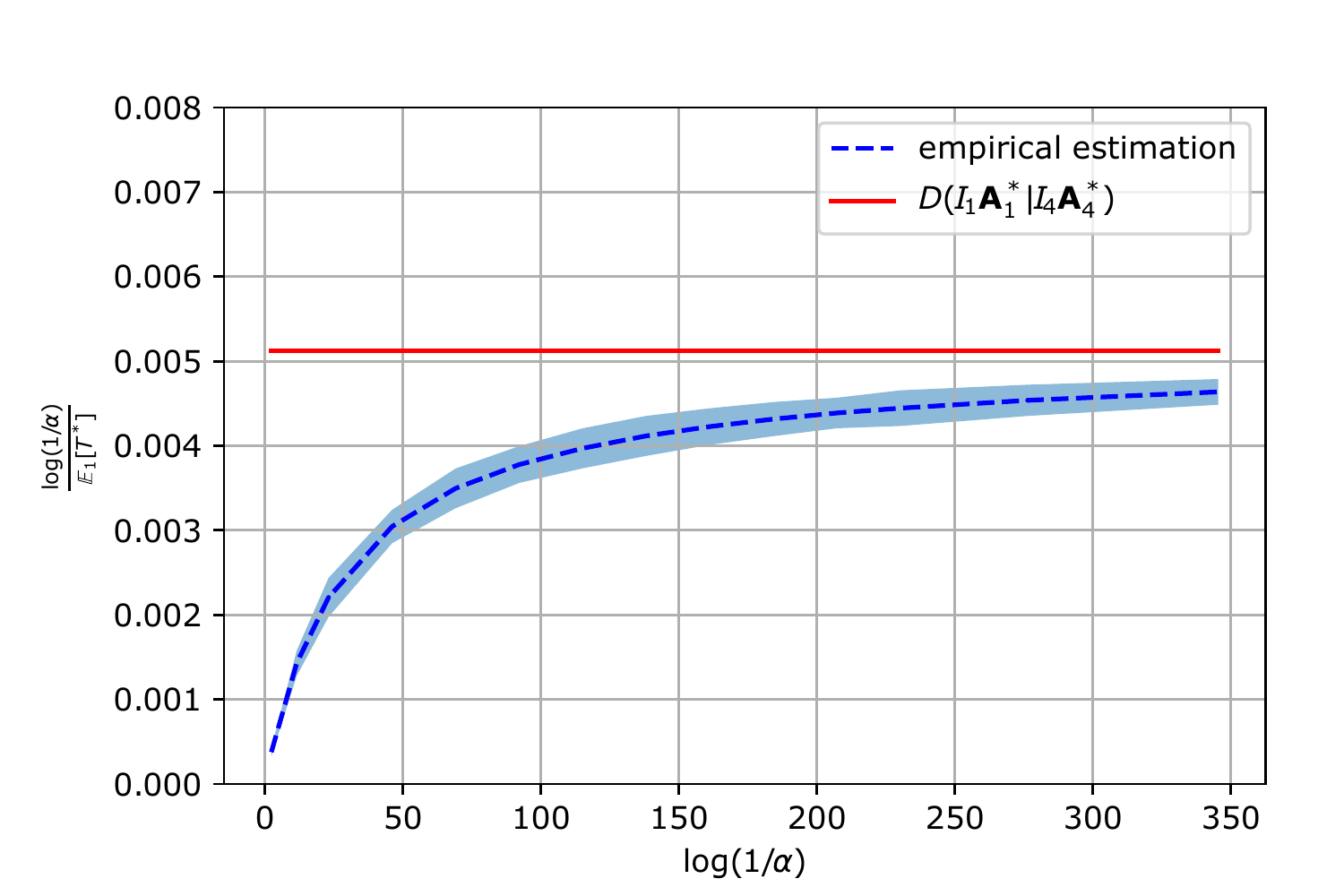}
    \caption{The change of the payoff function as $\alpha\to 0^+$  for  digit $1$.}
    \label{fig:mnist1}
\end{figure}

\section{Conclusion}
\label{sec:con}

In this work, we consider the $M$-ary sequential adversarial hypothesis testing problem. Different from the traditional $M$-ary sequential test, in this problem, an adversary is active and tries to perturbed the distributions of observed samples. Our objective is to obtain a pair of strategies for the adversary and the decision maker, in  which no party can increase its payoff by unilaterally changing its strategies, i.e., we wish to find the Nash equilibrium. In this paper, we obtain a pair of strategies at which the \emph{asymptotic Nash equilibrium} is attained. The adversary's strategy in the asymptotic Nash equilibrium is the transition matrices that minimize the Kullback--Leibler divergence between perturbed distributions, and the decision maker's strategy at the asymptotic Nash equilibrium is analogous to the sequential version of Hoeffding's test~\cite{hoeffding1965asymptotically}. 

In the future, several directions could be considered. First, in this paper, when consider the case that the adversary is not aware the underlying distribution of observed samples, the achievable and converse bounds do not match. This means that the pair of strategies we propose can not achieve the  Nash Equilibrium. We have endeavored to solve this problem but failed to find the pair of strategies attaining the  Nash equilibrium. Thus, we could consider to find the pair of strategies attaining the \emph{mixed Nash equilibrium}~\cite{osborne2004introduction}.  Second, as we only obtain the asymptotic Nash equilibrium when $\alpha\to 0^+$, one extension of our work is to consider the non-asymptotic Nash equilibrium for some fixed $\alpha\in (0,1)$. 
Third, we could also consider the case where the distribution of each hypothesis is unknown and we only access them through training sequences of each hypothesis. This case is analogue to sequential adversarial classification problem.

\appendix
\renewcommand{\thesubsection}{\Alph{subsection}}
\subsection{Proof of Lemma~\ref{lem:uniform}}
\label{sec:lemma5}
From the definitions of $T^*$ and $T_1$, we can see that $T_1\geq T^*$. Similar to the proof of Eqn.~\eqref{eqn:stoptime1}, we also can prove that 
\begin{equation}
 \lim_{\alpha\to 0^+}  \frac{\bbE_1[T_1]}{\log(1/\alpha)}= \frac{1}{\min_{j\neq 1}\big[\min_{\bA_j\in\calA_j} D(P_1\tilde{\bA}_1\|P_j\bA_j)\big]}. \label{eqn:conv}
\end{equation}
Now we want to show that the convergence above is uniform on $\calA_1$, which allows us to establish~\eqref{eqn:uniform}.

According to the definition of $T_1$, we have
\begin{align*}
&\min_{j\neq 1}\Big[\min_{\bA_j\in\calA_j}D(\hat{Q}_{Y^{T_1}}\|P_j\bA_j)\Big]\\
&\geq \frac{\log{(\frac{1}{\alpha})}}{T_1}+\frac{1}{T_1^{\zeta}}+\frac{|\calX|\log(T_1+1)+\log(M-1)}{T_1},
\end{align*}
and 
\begin{align*}
&\min_{j\neq 1}\Big[\min_{\bA_j\in\calA_j}D(\hat{Q}_{Y^{T_1-1}}\|P_j\bA_j)\Big]\\
&\leq \frac{\log{(\frac{1}{\alpha})}}{T_1-1}+\frac{1}{(T_1-1)^{\zeta}}+\frac{|\calX|\log(T_1)+\log(M-1)}{T_1-1}.
\end{align*}
Then, we have that 
\begin{align*}
\bigg|\frac{\log(1/\alpha)}{T_1}-{\min_{j\neq 1}\Big[\min_{\bA_j\in\calA_j}D(\hat{Q}_{Y^{T_1}}\|P_j\bA_j)\Big]}\bigg|\leq \frac{c_0 }{T_1^{\zeta}},
\end{align*}
where $c_0$ does not depend on $\tilde{\bA}_1$.  Then, we define $$D_{T_1}:=\min_{j\neq 1} \min_{\bA_j\in\calA_j}D(\hat{Q}_{Y^{T_1}}\|P_j\bA_j) ,$$ and $$D_1:=\min_{j\neq 1} \min_{\bA_j\in\calA_j}D(P_1\tilde{\bA}_1\|P_j\bA_j) .$$ We have
\begin{align*}
&\left|\bbE_1\left[\frac{\log(1/\alpha)}{T_1}-{D_1}\right]\right|\\
&=\left|\bbE_1\left[\frac{\log(1/\alpha)}{T_1}-{D_{T_1}}+{D_{T_1}}-{D_1}\right]\right|\\
&\leq \bbE_1\left[\left|\frac{\log(1/\alpha)}{T_1}-{D_{T_1}}\right|\right]+\bbE_1\big[\left|{D_{T_1}}-{D_1}\right| \big]\\
&\leq  \bbE_1\bigg[\frac{c_0}{T_1^{\zeta}}\bigg]+\bbE_1\big[\left|{D_{T_1}}-{D_1}\right| \big]
\end{align*}
Define $c_1:=\min_{j\neq 1}\big[\min_{{\bA}_j\in\calA_j}(-\log\min_{y\in\calX}{Q}_j(y))\big]$. For the first term, because
\begin{align*}
    &P_1(T_1\leq n)\\
    &\leq P_1\left(T_1\min_{j\neq 1}\Big[\min_{\bA_j\in\calA_j}D(\hat{Q}_{Y^{T_1}}\|P_j\bA_j)\Big]\!\geq\!\log\Big(\frac{1}{\alpha}\Big),T_1\leq n\right)\\
    &\leq P_1(c_1n>\log(1/\alpha),T_1\leq n)\\
    &=0,\qquad\forall \, n<\frac{\log(1/\alpha)}{c_1},
\end{align*}
we have that
\begin{align*}
P_1\left(T_1<\frac{\log(1/\alpha)}{c_1}\right)=0,
\end{align*}
This means that
\begin{align*}
T_1\geq \frac{\log(1/\alpha)}{c_1},\quad\mbox{a.s.}
\end{align*}
Thus,
\begin{align}
\label{eqn:bound1}
 \bbE_1\bigg[\frac{1}{T_1^{\zeta}}\bigg]\leq  \bigg( \frac{\log(1/\alpha)}{c_1}\bigg)^{-\zeta},
\end{align}
where $c_1$ does not depend on $\tilde{\bA}_i$. For the second term,
we define $c_2:=-\log\min_{\tilde{\bA}_i\in\calA_i}\min_{y\in\calX}\tilde{Q}_i(y)$. Let $\veps$ be an arbitrary fixed positive number.  Then we have that~\eqref{eqn:mid2} (on the top of next page),
\begin{figure*}
\begin{align}
\bbE_1\left[\left|{D_{T_1}}-{D_1}\right| \right]
&\stackrel{(a)}{\le} \bbE_1\bigg[ D(\hat{Q}_{Y^{T_1}}\|P_1\tilde{\bA}_1)+\max_{j\neq 1}\max_{\bA_j\in\calA_j}\sum_{a\in\calX}\bigg|(\hat{Q}_{Y^{T_1}}(a)-\tilde{Q}_1(a))\log\frac{\tilde{Q}_1(a)}{Q_j(a)}\bigg|\bigg]\notag\\
&\leq  \bbE_1\left[ D(\hat{Q}_{Y^{T_1}}\|P_1\tilde{\bA}_1)\right]+c_1|\calX|\bbE_1\left[ \sum_{a\in\calX}|\hat{Q}_{Y^{T_1}}(a)-\tilde{Q}_1(a)|\right]\notag\\
&\overset{(b)}{\leq } \bbE_1\left[ D(\hat{Q}_{Y^{T_1}}\|P_1\tilde{\bA}_1)\right]+c_3\bbE_1\left[\sqrt{D(\hat{Q}_{Y^{T_1}}\|P_1\tilde{\bA}_1)}\right]\notag\\
&=\bbE_1\left[ D(\hat{Q}_{Y^{T_1}}\|P_1\tilde{\bA}_1)\Big|D(\hat{Q}_{Y^{T_1}}\|P_1\tilde{\bA}_1)\geq\epsilon\right] P_1\left(D(\hat{Q}_{Y^{T_1}}\|P_1\tilde{\bA}_1)\geq\epsilon\right)\notag\\
&\qquad+P_1\left(D(\hat{Q}_{Y^{T_1}}\|P_1\tilde{\bA}_1)<\epsilon\right)\bbE_1\left[ D(\hat{Q}_{Y^{T_1}}\|P_1\tilde{\bA}_1)\Big|D(\hat{Q}_{Y^{T_1}}\|P_1\tilde{\bA}_1)<\epsilon\right]\notag\\
&\qquad+c_3\bbE_1\left[\sqrt{D(\hat{Q}_{Y^{T_1}}\|P_1\tilde{\bA}_1)}\Big|D(\hat{Q}_{Y^{T_1}}\|P_1\tilde{\bA}_1)\geq\epsilon\right] P_1\left(D(\hat{Q}_{Y^{T_1}}\|P_1\tilde{\bA}_1)\geq\epsilon\right)\notag\\
&\qquad+ c_3\bbE_1\left[\sqrt{D(\hat{Q}_{Y^{T_1}}\|P_1\tilde{\bA}_1)}\Big|D(\hat{Q}_{Y^{T_1}}\|P_1\tilde{\bA}_1)<\epsilon\right]P_1\left(D(\hat{Q}_{Y^{T_1}}\|P_1\tilde{\bA}_1)<\epsilon\right)\notag\\
&\leq \epsilon +\big(c_2|\calX|+\sqrt{c_2|\calX|}\big)P_1\left( D(\hat{Q}_{Y^{T_1}}\|P_1\tilde{\bA}_1)\geq\epsilon\right)+c_3\sqrt{\epsilon},\label{eqn:mid2}
\end{align}\hrulefill
\end{figure*}
where $(a)$ follows from that $|\min f(x)-\min g(x)|\le \max|f(x)-g(x)|$, $(b)$ follows from Pinsker's inequality~\cite[Lemma~11.6.1]{cover2006elements} and $c_2, c_3$ do not depend on $\tilde{\bA}_1$.
We also have 
\begin{align*}
P_1&\left( D(\hat{Q}_{Y^{T_1}}\|P_1\tilde{\bA}_1)\geq\epsilon\right)\\
&\leq \sum_{k\geq  \log(1/\alpha)/c_1} P_1\left( D(\hat{Q}_{Y^k}\|P_1\tilde{\bA}_1)\geq\epsilon\right)\\
&\leq  \sum_{k\geq \log(1/\alpha)/c_1} c_4e^{-k\epsilon}\\
&\leq c_5e^{-\frac{\log(1/\alpha)}{c_1}\epsilon},
\end{align*}
where $c_5$  depends only on $|\calX|$. Thus,
\begin{align}
\label{eqn:bound2}
&\bbE_1\big[\left|{D_{T_1}}-{D_1}\right| \big]\notag\\ 
&\quad \leq \epsilon +\big(c_2|\calX|+\sqrt{c_2|\calX|}\big)c_5 e^{-\frac{\log(1/\alpha)}{c_1}\epsilon}+c_3\sqrt{\epsilon}.
\end{align}
 Therefore, combining~\eqref{eqn:bound1} and~\eqref{eqn:bound2}, we have
\begin{align*}
\left|\bbE_1\left[\frac{\log(1/\alpha)}{T_1}-{D_1}\right]\right|&\leq  \epsilon +\big(c_2|\calX|+\sqrt{c_2|\calX|}\big)c_5 e^{-\frac{\log(1/\alpha)}{c_1}\epsilon}\\
&\quad+c_3\sqrt{\epsilon}+  c_0\bigg( \frac{\log(1/\alpha)}{c_1}\bigg)^{-\zeta}.
\end{align*}
As $c_i$ for $i = 0,1,\ldots, 5$ do not depend on $\tilde{\bA}_1$, the convergence in~\eqref{eqn:conv} is uniform over $\calA_1$. Now we show that the uniform convergence over $\calA_1$ also holds for $\big\{\frac{\bbE_1[T^*]}{\log(1/\alpha)}\big\}_{0<\alpha\leq 1}$ as $\alpha\to0^+$. 
For $i\in[M]$, let $\calB_i$ be the event that $T_i<T_1$. Note that $P_1\big(\cup_{i\not=1}\calB_i\big)$ is the error probability $\alpha_1$.
Conditioned on the events $\calB_i,i\neq 1$, we have $T^*<T_1$ and conditioned on the event $\calB_1$, we have $T^*=T_1$. Then
\begin{align}
\bbE_1[T^*]&=\bbE_1[T^*\mathbbm{1}_{ \calB_2\cup\cdots\cup \calB_M }]+\bbE_1[T^*\mathbbm{1}_{ \calB_{1} }]\notag\\
&=\bbE_1[T_1]+\bbE_1[(T^*-T_1)\mathbbm{1}_{ \calB_2\cup\cdots\cup \calB_M }]\notag\\
&\ge \bbE_1[T_1]-\bbE_1[T^*\mathbbm{1}_{ \calB_2\cup\cdots\cup \calB_M }]\label{eqn:lemma5}.
\end{align}
From Eqn.~\eqref{eqn:mid1} in the proof of uniform integrability, it follows that for the given $\veps>0$, there exists a finite constant $K>0$ that does not depend on $\tilde{\bA}_1$ such that for any $0<\alpha\le \alpha_{0}$ and any $(\tilde{\bA}_{1},\ldots,\tilde{\bA}_{M})$,
\begin{align*}
\bbE_1\bigg[\frac{T^*}{\log(1/\alpha)}\mathbbm{1}_{\{T^*(\alpha)/\log(1/\alpha)\ge K\}}\bigg]&\le \veps.
\end{align*} 
Therefore, we have that
\begin{align*}
&\bbE_1[T^*\mathbbm{1}_{ \calB_2\cup\cdots\cup \calB_M }]\\
&=\bbE_1\bigg[\frac{T^*}{\log(1/\alpha)}\mathbbm{1}_{\{T^*(\alpha)/\log(1/\alpha)\ge K\}}\mathbbm{1}_{ \calB_2\cup\cdots\cup \calB_M }\bigg]\log\Big(\frac{1}{\alpha}\Big)\\
&\quad+\bbE_1\bigg[\frac{T^*}{\log({1}/{\alpha})}\mathbbm{1}_{\{T^*(\alpha)/\log({1}/{\alpha})\le K\}}\mathbbm{1}_{ \calB_2\cup\cdots\cup \calB_M }\bigg]\log\Big(\frac{1}{\alpha}\Big)\\
&\le \veps\log\Big(\frac{1}{\alpha}\Big)+K\, P_{1}(\calB_2\cup\cdots\cup \calB_M)\log\Big(\frac{1}{\alpha}\Big)\\
&\overset{(a)}{\le }\veps\log\Big(\frac{1}{\alpha}\Big)+K\, \alpha\log\Big(\frac{1}{\alpha}\Big),
\end{align*}
where $(a)$ follows because $P_{1}(\calB_2\cup\cdots\cup \calB_M)$ is exactly the error probability $\alpha_1$ which  is upper bounded by $\alpha$. From~\eqref{eqn:lemma5}, we have
\begin{align}
\bbE_1\bigg[\frac{T_{1}}{\log(1/\alpha)}\bigg]&-\bbE_1\bigg[\frac{T^*}{\log(1/\alpha)}\bigg]\\
&\le \bbE_1\bigg[\frac{T^*}{\log(1/\alpha)}\mathbbm{1}_{ \calB_2\cup\cdots\cup \calB_M }\bigg]\\
&\le \veps+K\alpha,
\end{align} 
which, together with the arbitrariness of $\varepsilon$, implies that
\begin{equation}
\lim_{\alpha\to0^{+}}\sup_{\tilde{\bA}_{1}\in\mathcal{A}_{1}}\bigg(\bbE_1\bigg[\frac{T_{1}}{\log(1/\alpha)}\bigg]-\bbE_1\bigg[\frac{T^*}{\log(1/\alpha)}\bigg]\bigg)=0.\label{eqn:difference}
\end{equation}
Then it follows from the uniform convergence of $\bbE_1\big[\frac{T_{1}}{\log(1/\alpha)}\big]$ over $\mathcal{A}_{1}$ and~\eqref{eqn:difference} that
\begin{align*}
	\lim_{\alpha\to 0^+} \sup_{\tilde{\bA}_{1}\in\mathcal{A}_{1}}\bigg(\frac{\bbE_1[T^*]}{\log(1/\alpha)}-\frac{1}{D_1}\bigg)
=0,
\end{align*}
as desired. The arguments for other $i\in[M]$ proceed similarly.

\bibliographystyle{IEEEtran}
\bibliography{ref.bib}

\end{document}